\documentclass[11pt]{article}
\usepackage[latin9]{inputenc}
\usepackage{float}
\usepackage{amsmath}
\usepackage{amsthm}
\usepackage{amssymb}
\usepackage{graphicx}
\usepackage[authoryear]{natbib}

\makeatletter

\providecommand{\tabularnewline}{\\}
\floatstyle{ruled}
\newfloat{algorithm}{tbp}{loa}
\providecommand{\algorithmname}{Algorithm}
\floatname{algorithm}{\protect\algorithmname}

\theoremstyle{plain}
\newtheorem{thm}{\protect\theoremname}
\theoremstyle{definition}
\newtheorem{defn}[thm]{\protect\definitionname}
\theoremstyle{plain}
\newtheorem{lem}[thm]{\protect\lemmaname}

\usepackage{amsfonts}
\usepackage{amsthm}\usepackage{enumerate}\usepackage{epsfig}\usepackage{graphicx}\usepackage{ifthen}
\usepackage{latexsym}\usepackage{syntonly}\usepackage{rotating}
\usepackage{lscape}
\usepackage{graphicx}
\usepackage{color}
\usepackage[bottom]{footmisc}
\usepackage{longtable}
\providecommand{\U}[1]{\protect\rule{.1in}{.1in}}
\RequirePackage[colorlinks,citecolor=blue,urlcolor=black,linkcolor=red]{hyperref}
\newtheorem{assumption}{Assumption}\theoremstyle{definition}

\allowdisplaybreaks

\makeatother

\providecommand{\definitionname}{Definition}
\providecommand{\lemmaname}{Lemma}
\providecommand{\theoremname}{Theorem}

\begin{document}
\title{Reinforcement Learning in High-frequency Market Making}

\author{Yuheng Zheng\thanks{The author is grateful to Yacine A\"{\i}t-Sahalia for encouraging work on this topic and for the helpful suggestions provided thereafter.} \and Zihan Ding}
\date{Princeton University \\
\texttt{\{yuheng, zihand\}@princeton.edu}
}
\maketitle
\begin{abstract}
This paper establishes a new and comprehensive theoretical analysis for the application of reinforcement learning (RL) in high-frequency market making. We bridge the modern RL theory and the continuous-time statistical models in high-frequency financial economics. Different with most existing literature on methodological research about developing various RL methods for market making problem, our work is a pilot to provide the theoretical analysis. We target the effects of sampling frequency, and find an interesting tradeoff between error and complexity of RL algorithm when tweaking the values of the time increment $\Delta$ --- as $\Delta$ becomes smaller, the error will be smaller but the complexity will be larger. We also study the two-player case under the general-sum game framework and establish the convergence of Nash equilibrium to the continuous-time game equilibrium as $\Delta\rightarrow0$. The Nash Q-learning algorithm, which is an online multi-agent RL method, is applied to solve the equilibrium. Our theories are not only useful for practitioners to choose the sampling frequency, but also very general and applicable to other high-frequency financial decision making problems, e.g., optimal executions, as long as the time-discretization of a continuous-time markov decision process is adopted. Monte Carlo simulation evidence support all of our theories.

\end{abstract}
\tableofcontents{}

\section{Introduction}

Market making refers to the process where a trader, called a market
maker, posts quotes on both sides of the limit order book (LOB) to
provide liquidity and generate profit. The main sources of market
maker's profit come from capturing the bid-ask spread in the LOB,
and meanwhile, they need to avoid holding undesirably large positions
to control the inventory risk. It is natural to model the LOB dynamics
using stochastic processes, and formulate the market making problem
as a stochastic control problem since the target is clearly to maximize
the market maker's expected risk-adjust return.

The market making model used in our paper is based on the classical
inventory control framework (see, e.g., \citet{AmiMen1980HFT_MM},
\citet{HoStoll1981HFT_MM}, \citet{AveSto2008HFT_HJB}, \citet{GCFT2013HFT_MM_HJB_approx},
and \citet{cartea2014HFT_MM_HJB}). Under this framework, the key
differences among existing work are on the statistical models of LOB
dynamics they adopt. For instance, the seminal work \citet{AveSto2008HFT_HJB}
modeled the price by the Brownian motions and used the controlled
Poisson processes to model the market order flows, where the action
variable is the quoted price that influences the Poisson rates of
market orders; \citet{GCFT2013HFT_MM_HJB_approx} generalized the
model of \citet{AveSto2008HFT_HJB} by adding the drift term and market
impact term in the price dynamics, and made new contributions on the
closed-form approximation of the solution; \citet{cartea2014HFT_MM_HJB}
introduced predictable $\alpha$ in the dynamics of price and modeled
the market orders by multifactor mutually exciting processes to capture
the feedback effects.

The aforementioned papers studied the optimal quoting strategy of
a single market maker. In practice, the price competitions among different
market markers are also of interest. \citet{kyle1984market_game,kyle1985ECTA_market_game,kyle1989RES_market_game}
introduced informed traders, noisy traders and market makers under
the game-theoretical setup, and studied the price competition by explicitly
calculating the equilibrium. \citet{luo2021HJB_market_making_game}
considered multiple market makers who are symmetric to each other
and have incomplete information, and incorporated the inventory risks
in the competitive market making model. Recently, \citet{yacMeh2023HFTspeed}
studied the equilibrium between the high-frequency market maker, who
have both speed and informational advantages, with the low frequency
trader, and established lots of testable economic implications.

Besides various choices of market making models, an even more important
problem is how to solve the optimal strategies. A large amount of
previous work, including all of the aforementioned work, adopt the
classical Hamilton-Jacobi-Bellman (HJB) equation approach, which could
provide closed-form approximate or even exact solutions under some
settings. Though the HJB approach is attractive since it can possibly
provide analytical solutions and yield lots of insights, it requires
that the model of market dynamics is known, which is difficult in
reality.

Modern financial markets are increasingly electronic, and this electronification
has led to the emergence of big data. Recently, the data-driven machine
learning (ML) approaches become more and more popular in finance (see,
e.g, \citet{XiuKelly2023ML_finance_survey} for surveys), and lots
of works have applied ML to market making problem. Reinforcement learning
(RL, see, e.g., \citet{suttonBarto2018RLbook}), as a ML technique
to solve stochastic control problem, is a natural way to study market
making since the market making problem is essentially a stochastic
control problem. In general, most RL algorithms are designed to find
the optimal policies of the Markov decision process (MDP), which is
a discrete-time stochastic control problem, but RL can also solve
continuous-time problem after suitable time-discretizations of the
model. Compared with the HJB approach, the advantages of RL are that
it does not require the knowledge of the underlying model and is able
to learn the optimal policy directly from the data. Lots of work has
studied the applications of RL on market making (see, e.g., \citet{hamblyXuYang2023RL_finance_review}
and \citet{gavsperov2021RLMM_review} for surveys). Most of existing
works are about methodological research and focus on applying fancy
RL algorithms to various market data, however, there is a noticeable
lack of study on the theoretical analysis for applying RL algorithm
to market making.

Our paper is a pilot to provide a theoretical analysis for the application
of RL to high-frequency market making. Our focus is the effects of
the sampling frequency on the RL algorithm. We assume that the market
maker interacts with the LOB on the discrete-time grid $\{i\Delta\}_{i=0,1,2,...}$
and learns the policy using the standard RL algorithm, e.g., Q-learning.
Then we are concerned that, how does the sampling frequency $1/\Delta$
affect the optimal market making strategies and the performance of
RL algorithm? Does a smaller $\Delta$, which results in a higher
frequency, always lead to better learning results? To answer these
questions, we construct a family of MDPs indexed by $\Delta$ to capture
the effects of different frequencies, and then analyze the properties
of RL algorithm under these MDPs. Our work is also related with high-frequency
financial econometrics (see, e.g., \citet{yacjacodHFBOOK}), where
the key target is the statistical estimators constructed by the discrete
data of continuous-time processes, but our focus is different and
is the learning algorithm in the decision making problem.

In our high-frequency market making setup, we define statistical metrics
to characterize the error and the complexity of RL algorithm, and
surprisingly, we find an interesting tradeoff between error and complexity
when tweaking the values of $\Delta$: As $\Delta$ becomes smaller,
the error will be smaller but the complexity will be larger. These
two metrics of RL algorithm are of both theoretical and realistic
importances. As we will rigorously define later, the error measures
the accuracy of our estimation of the expected profit as the RL algorithm
iterates, while the complexity measures the transaction costs in some
way because every iteration of RL algorithm is actually a quote of
market maker. In particular, our theoretical analysis is very general
and is applicable to any discretized continuous-time MDP, and thus
can be applied to study other high-frequency financial decision making
problems, e.g., optimal executions. For practitioners, our results
suggest that the choice of sampling frequency should be paid attention
and depend on which aspect of the algorithm is given priority to.

Besides the contribution to single-agent case, we establish, under
the game-theoretical framework where two market makers have price
competitions, the convergence for Nash equilibrium of the discretized
model as $\Delta\rightarrow0$, and show that the limiting equilibrium
point is identical to the Nash equilibrium of the continuous-time
game under uniqueness assumption. We apply the Nash Q-learning algorithm
(see, e.g., \citet{JMLR2003Nash_Q_learning}), which is a multi-agent
RL algorithm, to solve the equilibrium for the discretized model.
Our results not only provide insight from theoretical sides, but also
offer an efficient method to obtain an approximation of equilibrium
for continuous-time game. The applications of RL in the competitive
market making model has been studied in recent years (see, e.g., \citet{ganesh2019RL_multiMM_game,xiong2021MM_collusion,ardon2021RL_multiMM_simulator,cont2022RL_MM_multi_agent,cartea2022MM_game_tick_size,cartea2022QF_MM_game,han2022make_take_fee_RL,wang2023robust_RL_multi,vadori2024multiRL_OTC_simulation}).
To name a few examples, \citet{xiong2021MM_collusion,cont2022RL_MM_multi_agent}
applied the deep RL approach to study the stochastic game among different
market makers, where the competition is modeled as a Nash equilibrium
and the tacit collusion is described in terms of Pareto optima. \citet{ardon2021RL_multiMM_simulator}
applied multi-agent RL to establish market simulators and built a
new framework where the makers and takers learn simultaneously to
optimize their objectives. \citet{wang2023robust_RL_multi} modeled
the market maker and the adversary in a zero-sum game and applied
the adversarial RL to learn a strategy that is robust in different
adversarial environments. The market making model in \citet{cont2022RL_MM_multi_agent}
is similar to ours and their RL algorithm requires the knowledge of
the transition probability model, while the Nash Q-learning algorithm
we use is model-free and our focus is the convergence of discretized
model to the continuous-time model, which make our paper distinct
with theirs.

The rest of our paper is organized as follows. Section \ref{sec:continuous-time HFMM model}
sets up the high-frequency market making model. Section \ref{sec:discrete-time HFMM and converge}
contructs the time-discretization of the continuous-time model in
Section \ref{sec:continuous-time HFMM model}, and shows the convergence
of the discretized model as the time increment $\Delta$ goes to zero.
Section \ref{sec:Q-learning sample cplx} provides an upper bound
for the sample complexity of Q-learning under the discretized model,
and discusses the tradeoff between error and complexity in details.
Section \ref{sec:two-player case} is dedicated to the game theoretical
setup where two market makers compete with each other. Section \ref{sec:Numerical results}
conducts the numerical studies to validate our theories. Section \ref{sec:Conclusion}
concludes. All proofs are in the Appendix.

\section{High-frequency Market Making\label{sec:continuous-time HFMM model}}

In this section, we introduce our high-frequency market making model.
Before going into the details, we make a general description of our
model and explain the motivations. Our model is based on the continuous-time
MDP in \citet{AveSto2008HFT_HJB} (itself following partly \citet{HoStoll1981HFMM}),
where the market order flow is modeled by a controlled Poisson process.
The differences are that, first, we model the price dynamics using
a controlled continuous-time Markov chain on a finite space, and second,
following \citet{GCFT2013HFT_MM_HJB_approx}, we put a bound $N_{Y}$
to the inventory that the market maker is allowed to have. One common
reason for these two constraints is that, they make the state variable
of our MDP be on finite space so that it is convenient to analyze
the complexity of RL algorithm later. Besides this technical reason,
they are also economically sensible. In real LOB, the quoted prices
are on a discrete set where the distance between two adjacent price
levels is the tick size, and for a single security, the prices which
have active quotes and executions are usually within a finite range.
Thus, it is natural to assume that the prices are on a equidistant
finite space (see, e.g., \citet{ContStoikovTalreja10}). The bound
on inventory is a realistic restriction to control the inventory risk
caused by undesirable movements of price.

Denote by $\mathcal{M}_{0}:=(\mathcal{S},\mathcal{A},\mathcal{P},R,\gamma)$
the continuous-time MDP, where $\mathcal{S}$ is the state space,
$\mathcal{A}$ is the action space, $\mathcal{P}$ is the transition
probability kernel, $R$ is the reward function, and $\gamma$ is
the discounted factor. We describe these components as follows.

\paragraph{The state space $\mathcal{S}:=\mathcal{S}_{X}\times\mathcal{S}_{Y}$
and the state variable $S_{t}:=(X_{t},Y_{t})$.}

Denote by $\delta_{P}$ the tick size. Suppose that the limit orders
are quoted on the discrete grid $\mathcal{S}_{P}$ defined as
\[
\mathcal{S}_{P}:=\{0,\delta_{P},2\delta_{P},...,(N_{P}-1)\delta_{P},N_{P}\delta_{P}\},
\]
where $N_{P}\in\mathbb{N}$ is a fixed positive integer. Denote by
$X_{t}$ the mid-price. Since the best bid price and the best ask
price take values in $\mathcal{S}_{P}$, the mid-price which equals
their average should take values in the state space $\mathcal{S}_{X}$
defined as
\[
\mathcal{S}_{X}:=\left\{ \frac{k}{2}\delta_{P}\text{ }|\text{ }k=1,2,...,(2N_{P}-1)\right\} .
\]
Denote by $Y_{i}$ the signed number of inventory held by the agent.
Suppose that the state space $\mathcal{S}_{Y}$ of $Y_{t}$ is given
by
\[
\mathcal{S}_{Y}:=\{-N_{Y},-(N_{Y}-1),...,-2,-1,0,1,2,...,(N_{Y}-1),N_{Y}\},
\]
where $N_{Y}\in\mathbb{N}$ is a fixed positive integer. The state
variable of the MDP is $S_{t}:=(X_{t},Y_{t})$, i.e., the mid-price
of the asset and the inventory holding. The state space $\mathcal{S}:=\mathcal{S}_{X}\times\mathcal{S}_{Y}$
is finite and $|\mathcal{S}|=|\mathcal{S}_{X}||\mathcal{S}_{Y}|=(2N_{P}-1)(2N_{Y}+1)$.

\paragraph{The action space $\mathcal{A}:=\mathcal{S}_{P}\times\mathcal{S}_{P}$
and the action variable $a_{t}:=(p_{t}^{a},p_{t}^{b})$.}

At time $t$, the action variable are the quoted prices $p_{t}^{a}$
and $p_{t}^{b}$ for the limit sell order and limit buy order, respectively,
submitted by the agent. Suppose that the volume of every sell order
and every buy order is always one unit of the asset, and the agent
obey the following restrictions.
\begin{equation}
\begin{tabular}{|c|ccc|ccc|}
\hline  State variable \ensuremath{(X_{t},Y_{t})}  &   & \multicolumn{4}{c}{Available actions \ensuremath{a_{t}=(p_{t}^{a},p_{t}^{b})} } &  \\
\hline   &   &  \text{Can she put a sell order?}  &   &   &  \text{Can she put a buy order?}  &  \\
\hline  if \ensuremath{Y_{t}=-N_{Y}}  &   &  No  &   &   &  \text{Yes, at any price level }\ensuremath{p_{t}^{b}<X_{t}}  &  \\
\hline  if \ensuremath{|Y_{t}|\leq N_{Y}-1}  &   &  \text{Yes, at any price level }\ensuremath{p_{t}^{a}>X_{t}}  &   &   &  \text{Yes, at any price level }\ensuremath{p_{t}^{b}<X_{t}}  &  \\
\hline  if \ensuremath{Y_{t}=N_{Y}}  &   &  \text{Yes, at any price level }\ensuremath{p_{t}^{a}>X_{t}}  &   &   &  No  &  
\\\hline \end{tabular}\label{table:quote rule}
\end{equation}
Here, the quoted prices $p_{i}^{a}$ and $p_{i}^{b}$ must take values
in the set $\mathcal{S}_{P}=\{k\delta_{P}$ $|$ $k=0,1,...,N_{P}\}$.

\paragraph{The transition probability kernel $\mathcal{P}$, the reward function
$R$, and the value function $V_{0}^{\pi}(s)$.}

The mid-price $X_{t}$ is a continuous-time Markov chain on state
space $\mathcal{S}_{X}$ with the transition rate matrix (a.k.a. Q-matrix)
$Q_{X}(a)$ defined as
\begin{equation}
Q_{X}(a):=\left[\begin{array}{ccccc}
-\lambda_{1,2}(a) & \lambda_{1,2}(a) & 0 & 0 & 0\\
\lambda_{2,1}(a) & -(\lambda_{2,1}(a)+\lambda_{2,3}(a)) & \lambda_{2,3}(a) & 0 & 0\\
0 & ... & ... & ... & 0\\
0 & 0 & ... & ... & ...\\
0 & 0 & 0 & \lambda_{|\mathcal{S}_{X}|,|\mathcal{S}_{X}|-1}(a) & -\lambda_{|\mathcal{S}_{X}|,|\mathcal{S}_{X}|-1}(a)
\end{array}\right],\label{X Q-matrix continuous}
\end{equation}
where the dependence of $\lambda_{k,l}(a)$ on the action variable
$a$ models the market impacts of the limit orders from the agents.
We propose the following assumption on the rate function $\lambda_{k,l}$.

\begin{assumption}\label{Assump_Q_rate_lambda}

For the mid-price, the transition rate functions $\lambda_{k,l}(a)$
satisfy that, there exists a constant $C_{\lambda}>0$ such that $0<\lambda_{k,l}(a)<C_{\lambda}$
for any $k,l$ and $a\in\mathcal{A}$.

\end{assumption}

Denote by $N_{t}^{a}$ (resp. $N_{t}^{b}$) a controlled Poisson process
with intensity $\lambda(|p_{t}^{a}-X_{t}|)$ (resp. $\lambda(|p_{t}^{b}-X_{t}|)$),
which models the coming flow of the market buy (resp. sell) order
at time $t$. The function $\lambda$ should be a monotonically decreasing
function, and following \citet{AveSto2008HFT_HJB}, our choice is
$\lambda(d):=\alpha\exp(-\kappa d)$ with the parameters $\alpha,\kappa>0$.
As mentioned in \citet{AveSto2008HFT_HJB}, this function form is
motivated from stylized facts of LOB, and the parameters $\alpha$
and $\kappa$ characterize statistically the liquidity of the security.

The inventory value $Y_{t}$ satisfies 
\[
Y_{t}=-N_{t}^{a}+N_{t}^{b},
\]
and the initial values are given by $Y_{t}=N_{t}^{a}=N_{t}^{b}=0$.
The value function to be maximized is given by
\[
V_{0}^{\pi}(s):=E\left[\left.\int_{0}^{+\infty}e^{-\gamma t}dR_{t}(S_{t},a_{t})\right\vert S_{0}=s\right],
\]
where $\gamma$ is the discounted factor, and the dynamics of the
running reward $R_{t}(S_{t},a_{t})$ is given by 
\begin{align*}
dR_{t}(S_{t},a_{t}) & =\underset{\text{profit from the sell order}}{\underbrace{1_{\{Y_{t}>-N_{Y}\}}\cdot(p_{t}^{a}-X_{t}-c)dN_{t}^{a}}}+\underset{\text{profit from the buy order}}{\underbrace{1_{\{Y_{t}<N_{Y}\}}\cdot(X_{t}-p_{t}^{b}-c)dN_{t}^{b}}}\\
 & +\underset{\text{change of inventory value}}{\underbrace{Y_{t}dX_{t}}}-\underset{\text{penalty of inventory holding}}{\underbrace{\psi(Y_{t})dt}},
\end{align*}
Here, $c>0$ is a constant representing the transaction cost, the
function $\psi$ is given by $\psi(y):=\phi y^{2}$ and the parameter
$\phi$ measures the risk averse level to the inventory risk. The
formulation of the reward $R_{t}(S_{t},a_{t})$ implies that, whenever
the agent is allowed to put a sell or buy order, she will always put
this order, and her previous orders that have not been executed will
not influence the execution probability of her current orders.

\section{Time-discretization and Convergence\label{sec:discrete-time HFMM and converge}}

In this section, we first introduce the discrete-time market making
model, which is a discrete-time approximation of the continuous-time
model in Section \ref{sec:continuous-time HFMM model}, and then we
show this discretized model is reasonable by proving a set of convergence
results. The discrete-time market making model we introduce here will
play a key role in studying the effects of sampling frequency later.

\subsection{Discrete-time model\label{sec:discrete-time HFMM model}}

In reality, the market maker can only interact with the market in
discrete-time grid. Thus, it is natural to study the discrete-time
approximation of the continuous-time MDP introduced in Section \ref{sec:continuous-time HFMM model}.
For any time increment $\Delta>0$, we consider the discrete-time
MDP $\mathcal{M}_{\Delta}:=(\mathcal{S},\mathcal{A},\mathcal{P}_{\Delta},R_{\Delta},e^{-\gamma\Delta})$
defined as follows. Here, $\mathcal{S}$ and $\mathcal{A}$ are the
state space and the action space, respectively, and they are the same
with those for $\mathcal{M}_{0}=(\mathcal{S},\mathcal{A},\mathcal{P},R,\gamma)$;
$\mathcal{P}_{\Delta}$, $R_{\Delta}$, and $e^{-\gamma\Delta}$ are
the transition probability kernel, reward function, and discounted
factor, respectively.

The state variable is given by $S_{i}:=(X_{i},Y_{i})$, where $X_{i}$
is the mid-price and $Y_{i}$ is the signed number of inventory held
by the agent at time $t=i\Delta$. The action variable $a_{i}=(p_{i}^{a},p_{i}^{b})$
is the quoted prices at time $t=i\Delta$. We assume that the agent
follows the same rules as we described in Table (\ref{table:quote rule})
based on the mid-price $X_{i}$ and the inventory value $Y_{i}$.
Then, suppose that the mid-price $X_{i}$ is a Markov chain on state
space $\mathcal{S}_{X}$, and the law of $X_{i+1}$ at the $(i+1)$th
stage is characterized by the transition probability matrix $P_{X}(\Delta|a)$
defined as
\begingroup
\small
\begin{equation}
P_{X}(\Delta|a):=\left[\begin{array}{ccccc}
1-\lambda_{1,2}(a)\Delta & \lambda_{1,2}(a)\Delta & 0 & 0 & 0\\
\lambda_{2,1}(a)\Delta & 1-(\lambda_{2,1}(a)+\lambda_{2,3}(a))\Delta & \lambda_{2,3}(a)\Delta & 0 & 0\\
0 & ... & ... & ... & 0\\
0 & 0 & ... & ... & ...\\
0 & 0 & 0 & \lambda_{|\mathcal{S}_{X}|,|\mathcal{S}_{X}|-1}(a)\Delta & 1-\lambda_{|\mathcal{S}_{X}|,|\mathcal{S}_{X}|-1}(a)\Delta
\end{array}\right],\label{X transition prob matrix discrete}
\end{equation}
\endgroup
where $\lambda_{k,l}(a)$ are the same as those of Q-matrix $Q_{X}(a)$
in (\ref{X Q-matrix continuous}). Suppose that, conditioning on $(S_{i},a_{i})$,
$n_{i}^{a}$ and $n_{i}^{b}$ are Bernouli random variables
\begin{align}
P(n_{i}^{a} & =1|S_{i},a_{i})=1-P(n_{i}^{a}=0|S_{i},a_{i})=\lambda(|p_{i}^{a}-X_{i}|)\Delta,\label{ask fill prob discrete}\\
P(n_{i}^{b} & =1|S_{i},a_{i})=1-P(n_{i}^{b}=0|S_{i},a_{i})=\lambda(|p_{i}^{b}-X_{i}|)\Delta,\label{bid fill prob discrete}
\end{align}
where $\lambda(d)=\alpha\exp(-\kappa d)$ is the same as that in the
intensity function of Poisson processes $N_{t}^{a}$ and $N_{t}^{b}$
in the continuous-time MDP $\mathcal{M}_{0}$, and $P(n_{i}^{a}=1|S_{i},a_{i})$
(resp. $P(n_{i}^{b}=1|S_{i},a_{i})$) is the probability that the
sell (resp. buy) limit order is executed. Then the inventory holding
$Y_{i+1}$ at the time $(i+1)\Delta$ is given by
\[
Y_{i+1}=Y_{i}-n_{i}^{a}+n_{i}^{b}.
\]

The value function to be maximized is given by
\[
V_{\Delta}^{\pi}(s):=E\left[\left.\sum_{i=0}^{+\infty}e^{-i\gamma\Delta}R_{\Delta}(S_{i},a_{i})\right\vert S_{0}=s\right],
\]
where, $e^{-\gamma\Delta}$ is the discounted factor, and the reward
function $R_{\Delta}(S_{i},a_{i})$ is given by
\begin{align*}
R_{\Delta}(S_{i},a_{i}) & :=\underset{\text{profit from the sell order}}{\underbrace{(p_{i}^{a}-X_{i}-c)n_{i}^{a}\cdot1_{\{Y_{i}>-N_{Y}\}}}}+\underset{\text{profit from the buy order}}{\underbrace{(X_{i}-p_{i}^{b}-c)n_{i}^{b}\cdot1_{\{Y_{i}<N_{Y}\}}}}\\
 & +\underset{\text{change of inventory value}}{\underbrace{(X_{i+1}-X_{i})Y_{i}}}-\underset{\text{penalty of inventory holding}}{\underbrace{\psi(Y_{i})\Delta}}.
\end{align*}
Here, $c>0$ is the transaction cost and $\psi(y):=\phi y^{2}$ is
the same as that appearing in the continuous-time reward function.
The $i$th stage of the MDP $\mathcal{M}_{\Delta}$ models what happens
at the physical time $t=i\Delta$. The formulation of the reward $R_{\Delta}(S_{i},a_{i})$
implies that, whenever the agent is allowed to put a sell or buy order,
she will always put this order, and her previous orders that have
not been executed will not influence the execution of her current
orders. All interactions between the agent and the market happen on
the discrete-time grids $\{i\Delta\}_{i=0,1,2,\ldots}$.

Admittedly, the MDP $\mathcal{M}_{\Delta}$ is not the exact time-discretization
of MDP $\mathcal{M}_{0}$. However, the above formulation is not only
sensible owing to the convergence results as shown in the next section,
but it is also convenient for our theoretical analysis later.

\subsection{Convergence of discrete-time MDP $\mathcal{M}_{\Delta}$ as $\Delta\rightarrow0$}

In this section, we show the convergence of the discrete-time MDP
$\mathcal{M}_{\Delta}$ to the continuous-time MDP $\mathcal{M}_{0}$
as $\Delta\rightarrow0$ in a suitable sense. As a preparation, we
define some relevant notions. For the continuous-time MDP $\mathcal{M}_{0}$,
following Section III.9 in \citet{FlemingSoner2006book_MDP}, we consider
the set $\mathcal{U}_{0}$ of admissible policies satisfying that
the action $a_{t}$ taken at time $t$ is $\mathcal{F}_{t}$-measurable.
In our model, $\mathcal{F}_{t}$ is the $\sigma$-algebra generated
by the history $\{S_{t_{1}},a_{t_{1}}\}_{0\leq t_{1}\leq t}$. The
optimal value function for $\mathcal{M}_{0}$ is defined as $V_{0}^{*}(s)=\sup_{\pi\in\mathcal{U}_{0}}V_{0}^{\pi}(s)$.
For the discrete-time MDP $\mathcal{M}_{\Delta}$, following Section
1.7 in \citet{gihSkorohod2012controlled}, we consider the set $\mathcal{U}_{\Delta}$
of admissible policies satisfying that the action $a_{i}$ taken at
time $i$ is $\mathcal{F}_{i}$-measurable. In our model, $\mathcal{F}_{i}$
is the $\sigma$-algebra generated by the history $\{S_{k}\}_{k=0,1,\ldots,i-1,i}$.
The optimal value function for $\mathcal{M}_{\Delta}$ is defined
as $V_{\Delta}^{*}(s)=\sup_{\pi\in\mathcal{U}_{\Delta}}V_{\Delta}^{\pi}(s)$.
Following \citet{kaku1971CTMDP_Bellman} and Section 1.7 in \citet{gihSkorohod2012controlled},
for the MDP $\mathcal{M}_{0}$ (resp. $\mathcal{M}_{\Delta}$), if
a policy can be expressed as $a_{t}=\pi(S_{t})$ (resp. $a_{i}=\pi(S_{i})$)
where $\pi:\mathcal{S}\rightarrow\mathcal{A}$ is a deterministic
mapping, then we say this policy is a stationary Markov policy and
we denote this policy as $\pi(\cdot)$.
\begin{thm}
\label{Thm:value discretize error}Under Assumption \ref{Assump_Q_rate_lambda},
there exist stationary Markov policies $\pi_{0}^{*}(\cdot)$ and $\pi_{\Delta}^{*}(\cdot)$,
such that the optimal value functions in the continuous-time MDP $\mathcal{M}_{0}$
and the discrete-time MDP $\mathcal{M}_{\Delta}$ are attained under
$\pi_{0}^{*}(\cdot)$ and $\pi_{\Delta}^{*}(\cdot)$, respectively,
i.e.,
\[
V_{0}^{*}(s)=V_{0}^{\pi_{0}^{*}}(s)\qquad\text{and}\qquad V_{\Delta}^{*}(s)=V_{\Delta}^{\pi_{\Delta}^{*}}(s).
\]
Moreover, assuming the uniqueness of the optimal policies in $\mathcal{M}_{0}$
and $\mathcal{M}_{\Delta}$, then we have that, there exists $\Delta_{0}$
such that, for any $\Delta\in(0,\Delta_{0})$, it holds:

(i) The policies $\pi_{\Delta}^{*}(\cdot)$ and $\pi_{0}^{*}(\cdot)$
are the identical, i.e., $\pi_{\Delta}^{*}(s)=\pi_{0}^{*}(s)$ for
any $s\in\mathcal{S}$.

(ii) The bound for the optimal value functions is given by
\[
||V_{\Delta}^{*}(\cdot)-V_{0}^{*}(\cdot)||\leq C_{V}\Delta,
\]
where the norm is defined as $||V_{\Delta}^{*}(\cdot)-V_{0}^{*}(\cdot)||:=\max_{s\in\mathcal{S}}|V_{\Delta}^{*}(s)-V_{0}^{*}(s)|$,
and $C_{V}>0$ is a constant that does not depend on $\Delta$.
\end{thm}

Theorem \ref{Thm:value discretize error} shows the convergence of
approximated discrete-time MDP $\mathcal{M}_{\Delta}$ to the continuous-time
MDP $\mathcal{M}_{0}$, in terms of both the optimal policy and optimal
value functions. An interesting interpretation of (i) is that, since
the optimal policies are mappings from the finite state space $\mathcal{S}$
to the finite action space $\mathcal{A}$, we can exactly recover
the optimal strategy of continuous-time MDP using discrete-time approximated
MDP as long as $\Delta$ is smaller than a threshold.

Compared with the existing literature on the convergence of discrete-time
approximation of MDP (see, e.g., \citet{BenRobin1982discretizedMDPconverge}),
one of the advances of our results is that we explicitly give an $O(\Delta)$
upper bound for the distance between the optimal value functions of
$\mathcal{M}_{\Delta}$ and $\mathcal{M}_{0}$. To the best of our
knowledge, no existing result has established similar results.

\section{Sample Complexity\label{sec:Q-learning sample cplx}}
\subsection{Q-learning for single-player case}
In this section, we study the complexity of RL algorithm under the
discrete-time MDP $\mathcal{M}_{\Delta}$. Let us briefly review the
relevant concepts. For any policy $\pi$, the Q function under policy
$\pi$ is defined as
\[
Q_{\Delta}^{\pi}(s,a):=E[R_{\Delta}(S_{i},a_{i})+e^{-\gamma\Delta}V_{\Delta}^{\pi}(S_{i+1})|S_{i}=s,a=a_{i}],
\]
where $V_{\Delta}^{\pi}(s)$ is the value function under policy $\pi$,
and $R_{\Delta}(s,a)$ is the stochastic reward function defined in
Section \ref{sec:discrete-time HFMM model}. The optimal Q function
is given by $Q_{\Delta}^{*}(s,a)=Q_{\Delta}^{\pi_{\Delta}^{*}}(s,a)$,
which is the Q function under the optimal policy $\pi_{\Delta}^{*}(\cdot)$.

The RL algorithm we will study is the Q-learning algorithm (see, e.g.,
Section 6.5 in \citet{suttonBarto2018RLbook}), which is one of the
most popular algorithms. To make our paper self-contained, we explain
how Q-learning works in our case as follows. In our model, since the
state-action space $\mathcal{S}\times\mathcal{A}$ is finite, the
Q-learning algorithm is running in a tabular case. Also, Denote by
$Q_{\Delta}^{(n)}(s,a)$ the Q function learned at the $n$th iteration,
and denote by $V_{\Delta}^{(n)}(s)=\max_{a\in\mathcal{A}}Q_{\Delta}^{(n)}(s,a)$
the optimal value function learned at the $n$th iteration. The update
in each iteration of Q-learning is given by
\[
Q_{\Delta}^{(n+1)}(s,a)=Q_{\Delta}^{(n)}(s,a)+\beta^{(n)}(s,a)(R_{\Delta}(s,a)+e^{-\gamma\Delta}\max_{a\in\mathcal{A}}Q_{\Delta}^{(n)}(s^{\prime},a)-Q_{\Delta}^{(n)}(s,a)),
\]
where $(s,a)$ is the state-action pair at the $n$th iteration, $s^{\prime}$
is the state at the $(n+1)$th iteration and $s^{\prime}$ is the
state reached from the state $s$ after taking action $a$, $R_{\Delta}(s,a)$
is the stochastic reward function defined in Section \ref{sec:discrete-time HFMM model}.
Following \citet{JMLR2003Qcomplexity}, we use a polynomial learning
rate $\beta^{(n)}(s,a)$ which is given by $\beta^{(n)}(s,a)=(N(s,a,n))^{-\omega}$,
where $\omega\in(\frac{1}{2},1)$ and $N(s,a,n)$ is the one plus
the number of times, until the $n$th iteration, that we visited the
state-action pair $(s,a)$. The initial value of the Q table is set
as $Q_{\Delta}^{(0)}(s,a)=C_{0}$ for some constant $C_{0}>0$.

For the exploration method used in Q-learning, we adopt the $\varepsilon$-greedy
method (see, e.g., \citet{li2012RLcomplexity}), which is one of the
most popular methods. In our case, the $\varepsilon$-greedy method
works as follows: at the $n$th iteration, when the current state
is $s$, we take the greedy policy $a^{*}=\text{arg}\max_{a\in\mathcal{A}}Q_{\Delta}^{(n)}(s,a)$
with probability $1-\varepsilon^{(n)}(s)$, otherwise, randomly take
an action $a\in\mathcal{A}$ with equal probabilities for each action.
Here, $\{\varepsilon^{(n)}(s)\}$ is the sequence of $\varepsilon$
and adapts to the state variable $s$. For the convenience of theoretical
analysis, the sequence of $\varepsilon$ we use decays to a nonzero
small value $\varepsilon_{0}$ rather than $0$. The details of $\{\varepsilon^{(n)}(s)\}$
are discussed in Section \ref{sec:Numerical results}.

The complexity measure we use is the sample complexity (see, e.g.,
\citet{li2012RLcomplexity}), which captures the exploration efficiency
of RL algorithm. Following \citet{JMLR2003Qcomplexity}, we define
the sample complexity as the number of iteration steps $n$ such that,
with probability at least $1-\delta$, it holds $||V_{\Delta}^{(n)}(\cdot)-V_{\Delta}^{*}(\cdot)||\leq\varepsilon_{V}$,
where the norm is the same as that in Theorem \ref{Thm:sample cplx bound Delta}
and is defined as $||V_{\Delta}^{(n)}(\cdot)-V_{\Delta}^{*}(\cdot)||:=\max_{s\in\mathcal{S}}|V_{\Delta}^{*}(s)-V_{0}^{*}(s)|$.
In other words, with high probability, we need at most $n$ iterations
to get an estimation of optimal value function within the accuracy
level $\varepsilon_{V}$. Our definition of sample complexity is slightly
different with that of \citet{JMLR2003Qcomplexity} in that, we use
the value function error while \citet{JMLR2003Qcomplexity} used the
Q function error.
\begin{thm}
\label{Thm:sample cplx bound Delta}Under Assumption \ref{Assump_Q_rate_lambda},
for the sample complexity of the Q-learning algorithm with $\varepsilon$-greedy
exploration in the discrete-time MDP $\mathcal{M}_{\Delta}$, we have
that, for any sufficiently small time increment $\Delta$ and error
level $\varepsilon_{V}$, with probability at least $1-\delta$, it
holds $||V_{\Delta}^{(n)}(\cdot)-V_{\Delta}^{*}(\cdot)||\leq\varepsilon_{V}$
as long as
\begin{align}
n & =\Omega\left(((|\mathcal{S}_{X}|+|\mathcal{S}_{Y}|)|\mathcal{A}|\varepsilon_{0}^{-1})^{3+\frac{1}{\omega}}\varepsilon_{V}^{-\frac{2}{\omega}}\gamma^{-\frac{4}{\omega}}\Delta^{6-\frac{2}{\omega}-(|\mathcal{S}_{X}|+|\mathcal{S}_{Y}|)(3+\frac{1}{\omega})}\right)\nonumber \\
 & +\Omega\left(((|\mathcal{S}_{X}|+|\mathcal{S}_{Y}|)|\mathcal{A}|\varepsilon_{0}^{-1})^{\frac{1}{1-\omega}}\gamma^{-\frac{1}{1-\omega}}\Delta^{(1-|\mathcal{S}_{X}|-|\mathcal{S}_{Y}|)\frac{1}{1-\omega}}\right),\label{sample cplx bound}
\end{align}
where $\Omega$ suppresses logarithmic factors of $\frac{1}{\Delta}$,
$\frac{1}{\delta}$, $\frac{1}{\varepsilon_{Q}}$, $\frac{1}{\gamma}$,
$|\mathcal{S}|$, and $|\mathcal{A}|$.
\end{thm}

In the high-frequency market making model setup, the sample complexity
is closely related with the transaction costs, since every iteration
of Q-learning corresponds to a quote of the market maker, which may
require a fee in some exchange. The reward functions of our models
in Sections \ref{sec:continuous-time HFMM model} and \ref{sec:discrete-time HFMM model}
do not include the transaction costs due to the convenience of technical
analysis. Thus, it is of both theoretical and realistic importance
to analyze this sample complexity upper bound, so that we can have
a better understanding on how the transaction costs vary with the
sampling frequency.

A notable and desirable result in Theorem \ref{Thm:sample cplx bound Delta}
is that, the sample complexity upper bound is polynomial, rather than
exponential, in all model parameters. In (\ref{sample cplx bound}),
the exponents of $\Delta$ are both negative, since $6-\frac{2}{\omega}-(|\mathcal{S}_{X}|+|\mathcal{S}_{Y}|)(3+\frac{1}{\omega})\leq6-\frac{2}{\omega}-2(3+\frac{1}{\omega})\leq-\frac{4}{\omega}<0$
and $(1-|\mathcal{S}_{X}|-|\mathcal{S}_{Y}|)\frac{1}{1-\omega}\leq-\frac{1}{1-\omega}<0$.
From this upper bound, we have the following interpretation: This
bound is decreasing in $\Delta$ and will grow to infinity as $\Delta$
goes to zero. So, the larger the quoting frequency $1/\Delta$, the
larger the bound of sample complexity. A related work is \citet{bayr2023approxRLdiffusion},
but their model is different with ours and their focus is to develop
learning algorithm for the controlled diffusion process without the
background and applications in finance.

\subsection{Tradeoff between learning error and sample complexity}

Based on the above theoretical results, we now discuss the tradeoff
between error and complexity mentioned in the introduction. The learning
error is characterized by $||V_{\Delta}^{(n)}(\cdot)-V_{0}^{*}(\cdot)||$,
where the groundtruth is the optimal value function $V_{0}^{*}(\cdot)$
of the continuous-time MDP $\mathcal{M}_{0}$ and its estimator is
the value function $V_{\Delta}^{(n)}(\cdot)$ learned at the $n$th
iteration when running the RL algorithm under the discrete-time MDP
$\mathcal{M}_{\Delta}$. The sample complexity is defined as discussed
before Theorem \ref{Thm:sample cplx bound Delta}. Combining the results
in Theorems \ref{Thm:value discretize error} and \ref{Thm:sample cplx bound Delta},
we have that, when the iteration number $n$ satisfies (\ref{sample cplx bound}),
the learning error can be bounded by
\[
||V_{\Delta}^{(n)}(\cdot)-V_{0}^{*}(\cdot)||\leq||V_{\Delta}^{(n)}(\cdot)-V_{\Delta}^{*}(\cdot)||+||V_{\Delta}^{*}(\cdot)-V_{0}^{*}(\cdot)||\leq\varepsilon_{V}+C_{V}\Delta.
\]
Regarding the upper bound $\varepsilon_{V}+C_{V}\Delta$ as a metric
of the learning error and the upper bound in (\ref{sample cplx bound})
as a metric of the sample complexity, we have the tradeoff as indicated
by the curves in Figure \ref{figure: error complexity tradeoff} when
the sampling frequency $1/\Delta$ changes and other parameters are
fixed.
\begin{figure}[htbp]
\centering%
\begin{tabular}{c}
\includegraphics[scale=0.8]{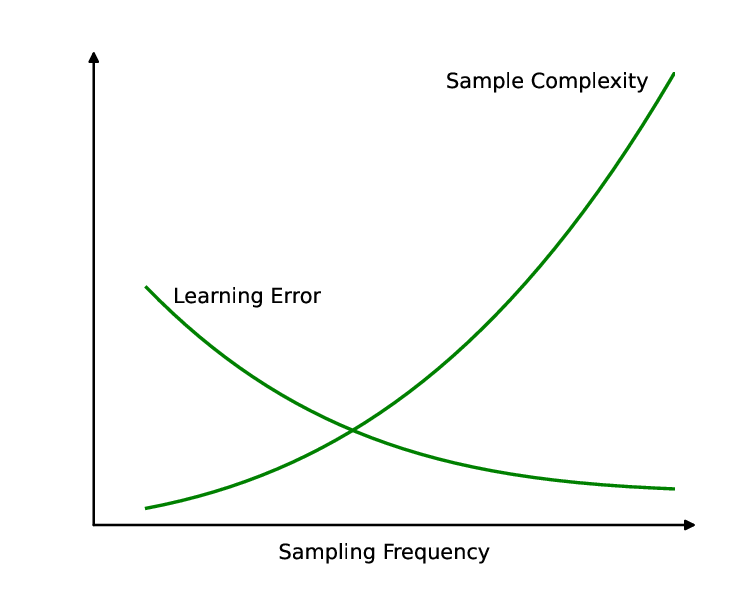}\tabularnewline
\end{tabular}\caption{The tradeoff between learning error and sample complexity\label{figure: error complexity tradeoff}}
\end{figure}

As the sampling frequency $1/\Delta$ increases, i.e., the time increment
$\Delta$ decreases, the learning error goes down while the sample
complexity goes up. For practitioners, our results suggest that, the
choice of sampling frequency should depend on which aspect of the
RL algorithm is primarily concerned. If the accuracy of the estimates
of expected profit is given the priority to, then a relatively high
sampling frequency is desirable; if the sample complexity that reflects
the transaction costs is given more concerns, then it is better to
use a relatively low sampling frequency.

\section{Two-player General-sum Setting\label{sec:two-player case}}

In this section, we study the price competition of two market makers
in a general-sum game setup which is an extention of the single-player
setup discussed previously. Our model is motivated by the game theoretical
framework in \citet{luo2021HJB_market_making_game} and \citet{cont2022RL_MM_multi_agent}.

\subsection{Continuous-time model and Nash equilibrium\label{sec:continuous MM game}}

Denote by $\text{MM}_{1}$ and $\text{MM}_{2}$ the two market makers.
Assume that they share the same state variable $S_{t}=X_{t}$, which
is the mid-price of the traded asset. Same with the single-player
setup in Section \ref{sec:continuous-time HFMM model}, the state
space of $X_{t}$ is $\mathcal{S}_{X}=\{\frac{k}{2}\delta_{P}\text{ }|\text{ }k=1,2,\ldots,(2N_{P}-1)\}$,
where $\delta_{P}$ is the tick size. At time $t$, the action variable
$a_{t}^{k}$ of $\text{MM}_{k}$ is the quoted prices $a_{t}^{k}=(p_{t}^{a,k},p_{t}^{b,k})$
of her limit sell order and limit buy order, for $k=1,2$. Both of
their limit orders are assumed to have one unit of the asset. The
mid-price $X_{t}$ is a controlled Markov chain, and its transition
rate matrix $Q_{X}(a_{t})$ (a.k.a. Q-matrix) at time $t$ is a function
of $a_{t}:=(a_{t}^{1},a_{t}^{2})\equiv(p_{t}^{a,1},p_{t}^{b,1},p_{t}^{a,2},p_{t}^{b,2})$,
where $Q_{X}$ is given by
\begin{equation}
Q_{X}(a):=\left[\begin{array}{ccccc}
-\lambda_{1,2}(a) & \lambda_{1,2}(a) & 0 & 0 & 0\\
\lambda_{2,1}(a) & -(\lambda_{2,1}(a)+\lambda_{2,3}(a)) & \lambda_{2,3}(a) & 0 & 0\\
0 & \ldots & \ldots & \ldots & 0\\
0 & 0 & \ldots & \ldots & \ldots\\
0 & 0 & 0 & \lambda_{|\mathcal{S}_{X}|,|\mathcal{S}_{X}|-1}(a) & -\lambda_{|\mathcal{S}_{X}|,|\mathcal{S}_{X}|-1}(a)
\end{array}\right],\label{Q matrix 2player}
\end{equation}
for any $a=(a^{1},a^{2})$. Here, the dependence of $Q_{X}(a)$ on
the action variable $a$ captures the price impacts of the limit orders
submitted by the market makers.

Next, we specify the execution probability functions that model the
intensity of the market orders. It is no longer as simple as $\lambda(|p_{t}^{a}-X_{t}|)$
in the single-player setup, because we need to incorporate the competition
between the two market makers now. Denote by $L_{t}^{a,k}=\Gamma^{a,k}(X_{t},p_{t}^{a,1},p_{t}^{a,2})$
(resp. $L_{t}^{b,k}=\Gamma^{b,k}(X_{t},p_{t}^{b,1},p_{t}^{b,2})$)
the Poisson rate of the buy (resp. ask) market order flow that executes
the ask (resp. buy) limit order of the market maker $\text{MM}_{k}$
for $k=1,2$. We set the functions $\Gamma^{a,k}$ and $\Gamma^{b,k}$
as 
\begin{equation}
\Gamma^{a,k}(x,p^{a,1},p^{a,2}):=\frac{\Upsilon^{-}(|p^{a,k}-x|)}{\Upsilon^{+}(|p^{a,k}-\min(p^{a,1},p^{a,2})|)}\text{ and }\Gamma^{b,k}(x,p^{b,1},p^{b,2}):=\frac{\Upsilon^{-}(|p^{b,k}-x|)}{\Upsilon^{+}(|p^{b,k}-\max(p^{b,1},p^{b,2})|)},\label{Gamma intensity function}
\end{equation}
where $\min(p^{a,1},p^{a,2})$ (resp. $\max(p^{b,1},p^{b,2})$) is
the best ask (resp. bid) prices from the two market makers. We impose
the following assumptions on $\Upsilon^{-}$ and $\Upsilon^{+}$.

\begin{assumption}\label{Assump_2player_intensity}

For the market order intensity functions $\Gamma^{a,k},\Gamma^{b,k}$,
their building blocks $\Upsilon^{-}$ and $\Upsilon^{+}$ satisfy
that: (i) $\Upsilon^{-}(d)$ (resp. $\Upsilon^{+}(d)$) is monotonically
decreasing (resp. increasing) function of $d$; (ii) There exists
a constant $C_{\Upsilon}>0$ such that $0<\Upsilon^{-}(d)<C_{\Upsilon}$
and $0<\Upsilon^{+}(d)<C_{\Upsilon}$ for any $d\geq0$; (iii) There
exists a constant $c^{+}>0$ such that $\Upsilon^{+}(d)>c^{+}$ for
any $d\geq0$.

\end{assumption}

The economic intuition behind $\Upsilon^{-}$ and $\Upsilon^{+}$
is that, the closer the quoted price is to the mid-price and the best
price, the more likely it is that the limit order will be executed
(see, e.g., \citet{luo2021HJB_market_making_game} and \citet{cont2022RL_MM_multi_agent}
for similar specifications of intensity functions). The lower bound
on $\Upsilon^{+}$ guarantees the uniformly boundness of the intensity
functions $\Gamma^{a,k}$ and $\Gamma^{b,k}$.

Denote by $N_{t}^{a,k}$ (resp. $N_{t}^{b,k}$) the controlled Poisson
process with intensity $L_{t}^{a,k}$ (resp. $L_{t}^{b,k}$), which
models the flow of the market buy (resp. sell) order at time $t$.
The value function of the market maker $\text{MM}_{k}$ is given by
\[
V_{0}^{k,\pi^{1},\pi^{2}}(s):=E\left[\left.\int_{0}^{+\infty}e^{-\gamma t}dR_{t}^{k}(S_{t},a_{t}^{1},a_{t}^{2})\right\vert S_{0}=s\right],
\]
for $k=1,2$, where $\gamma$ is the discounted factor, and the running
reward $R_{t}^{k}(S_{t},a_{t}^{1},a_{t}^{2})$ satisfies
\[
dR_{t}^{k}(S_{t},a_{t}^{1},a_{t}^{2})=\underset{\text{profit from the sell order}}{\underbrace{(p_{t}^{a,k}-X_{t}-c)dN_{t}^{a,k}}}+\underset{\text{profit from the buy order}}{\underbrace{(X_{t}-p_{t}^{b,k}-c)dN_{t}^{b,k}}.}
\]
Here, $c>0$ is the transaction cost. We focus on the price competition
and simplify the model by ignoring the inventory control in the reward
function.

In this two-player game theoretical framework, the target of every
market maker is to maximize the value function of herself. Our setup
belongs to the class of noncooperative stochastic game (e.g., \citet{FilarVrieze2012Book_MDP_game}),
which means that the two players optimize their individual target
and cannot form an enforceable agreement on joint actions. At any
time $t$, the agents choose their actions simultaneously and independently.
A commonly studied optimality condition is the Nash equilibrium (see,
e.g., \citet{nash1951AOMequili}). Before given the definition, we
introduce some relevant notions first. Denote by $\mathcal{P}(\mathcal{A})$
is the space of probability measures on the action space $\mathcal{A}$.
For the continuous-time game $\mathcal{G}_{0}$, following \citet{GuoHernz2005continuousMDPnonzero_sum_game},
we consider the randomized Markov strategies $\Pi_{M}^{1}$, and its
subset $\Pi_{s}^{1}$ called the stationary strategies. Formally,
$\Pi_{M}^{1}$ is defined as the family of strategies satisfying that,
for any $t\geq0$, there exists a mapping $\pi_{t}^{1}:\mathcal{S}\rightarrow\mathcal{P}(\mathcal{A})$,
such that for any $(s,a)\in\mathcal{S}\times\mathcal{A}$, the MM1
takes the action $a$ with probability $\pi_{t}^{1}(a|s)$ at time
$t$ when the state variable is $S_{t}=s$; $\Pi_{s}^{1}$ is the
subset of $\Pi_{M}^{1}$ satisfying that, there exists a mapping $\pi^{1}:\mathcal{S}\rightarrow\mathcal{P}(\mathcal{A})$
such that $\pi_{t}^{1}(a|s)=\pi^{1}(a|s)$ for any $t\geq0$ and $(s,a)\in\mathcal{S}\times\mathcal{A}$.
The strategy sets $\Pi_{M}^{2}$ and $\Pi_{s}^{2}$ for $\text{MM}_{2}$
are defined in the same manner.

For each pair $(\pi^{1},\pi^{2})=\{(\pi_{t}^{1},\pi_{t}^{2})\}_{t\geq0}\in\Pi_{M}^{1}\times\Pi_{M}^{2}$,
the $(i,j)$-th entry of Q-matrix of the controlled Markov chain $X_{t}$
is defined as $q_{t,ij}(\pi^{1},\pi^{2})=\sum_{a^{1}\in\mathcal{A},a^{2}\in\mathcal{A}}Q_{X,ij}(a^{1},a^{2})\pi_{t}^{1}(a^{1}|s)\pi_{t}^{2}(a^{2}|s)$,
where $Q_{X,ij}(a^{1},a^{2})$ is the $(i,j)$-th entry of the Q-matrix
$Q_{X}(a^{1},a^{2})$ defined in (\ref{Q matrix 2player}). To guarantee
the existence of the process $X_{t}$, we restrict the admissible
strategy sets in the classes $\Pi^{1}$ and $\Pi^{2}$ defined as
follows: $\Pi^{1}:=\{\pi^{1}\in\Pi_{M}^{1}:q_{t,ij}(\pi^{1},\pi^{2})\text{ is continuous in }t\text{ for any \ensuremath{i,j}}\text{ and }\pi^{2}\in\Pi_{M}^{2}\}$
and $\Pi^{2}:=\{\pi^{2}\in\Pi_{M}^{2}:q_{t,ij}(\pi^{1},\pi^{2})\text{ is continuous in }t\text{ for any \ensuremath{i,j}}\text{ and }\pi^{1}\in\Pi_{M}^{1}\}$.
By the uniformly boundness of $Q_{X,ij}(a^{1},a^{2})$, we have that
$\Pi_{s}^{k}\subseteq\Pi^{k}\subseteq\Pi_{M}^{k}$ for $k=1,2$.

Following Definition 4.1 in \citet{GuoHernz2005continuousMDPnonzero_sum_game},
we define the optimality condition for the continuous-time game $\mathcal{G}_{0}$
as follows.
\begin{defn}
A pair of strategies $(\pi_{0}^{1,*},\pi_{0}^{2,*})\in\Pi^{1}\times\Pi^{2}$
is called a Nash equilibrium if
\[
V_{0}^{1,\pi_{0}^{1,*},\pi_{0}^{2,*}}(s)\geq V_{0}^{1,\pi^{1},\pi_{0}^{2,*}}(s)\qquad\text{and}\qquad V_{0}^{2,\pi_{0}^{1,*},\pi_{0}^{2,*}}(s)\geq V_{0}^{2,\pi_{0}^{1,*},\pi^{2}}(s),
\]
for any $(\pi^{1},\pi^{2})\in\Pi^{1}\times\Pi^{2}$ and $s\in\mathcal{S}$.

The interpretation is that, when the two players take the strategies
$(\pi_{0}^{1,*},\pi_{0}^{2,*})$ in the Nash equilibrium, none of
them has the motivation to deviate from their strategy unilaterally.
In our market making setup, the Nash equilibrium point means that
none of the market maker wants to further adjust her quoted price
in this price competition with the other.
\end{defn}

\subsection{Time-discretization and convergence of equilibrium\label{sec:discrete MM game}}

In this section, we introduce the discrete-time stochastic game model
$\mathcal{G}_{\Delta}$, which serves as a discrete-time approximation
of the continuous-time model $\mathcal{G}_{0}$ in the last section.
Then, we show that the Nash equilibrium point of the discrete-time
model converges to that of the continuous-time model as the sampling
time increment $\Delta$ goes to zero.

Under $\mathcal{G}_{\Delta}$, the two market makers share the same
state variable $S_{i}=X_{i}$, which is the mid-price of the traded
asset at time $t=i\Delta$. The action variable of the market maker
$\text{MM}_{k}$ is $a_{i}^{k}=(p_{i}^{a,k},p_{i}^{b,k})$ of her
limit sell order and limit buy order at time $t=i\Delta$, for $k=1,2$.
Both of their limit orders are assumed to have one unit of the asset.
The transition probability matrix that characterizes the conditional
distribution of $X_{i+1}$ given $X_{i}$ is set to be $P_{X}(\Delta|a^{1},a^{2})=I-Q_{X}(a^{1},a^{2})\Delta$,
where $I$ is the identity matrix and $Q_{X}(a^{1},a^{2})$ is the
Q-matrix defined in (\ref{Q matrix 2player}) for the continuous-time
model $\mathcal{G}_{0}$. The execution probability for the ask (resp.
buy) limit order of $\text{MM}_{k}$ quoted at time $t=i\Delta$ is
given by $\Gamma^{a,k}(X_{i},p_{i}^{a,1},p_{i}^{a,2})\Delta$ (resp.
$\Gamma^{b,k}(X_{i},p_{i}^{b,1},p_{i}^{b,2})\Delta$) for $k=1,2$,
where $\Gamma^{a,k},\Gamma^{b,k}$ are the intensity functions defined
in (\ref{Gamma intensity function}). In other words, if we denote
by $n_{i}^{a,k}$ (resp. $n_{i}^{b,k}$) the indicator random variable
such that $n_{i}^{a,k}=1$ (resp. $n_{i}^{b,k}=1$) means the ask
(resp. buy) limit order of $\text{MM}_{k}$ is executed, then we have
that
\begin{align}
P(n_{i}^{a,k}=1|S_{i},a_{i}^{1},a_{i}^{2}) & =1-P(n_{i}^{a,k}=0|S_{i},a_{i}^{1},a_{i}^{2})=\Gamma^{a,k}(X_{i},p_{i}^{a,1},p_{i}^{a,2})\Delta,\label{ask fill prob discrete 2player}\\
P(n_{i}^{b,k}=1|S_{i},a_{i}^{1},a_{i}^{2}) & =1-P(n_{i}^{b,k}=0|S_{i},a_{i}^{1},a_{i}^{2})=\Gamma^{b,k}(X_{i},p_{i}^{b,1},p_{i}^{b,2})\Delta.\label{bid fill prob discrete 2player}
\end{align}
The value function of the market maker $\text{MM}_{k}$ is given by
\[
V_{\Delta}^{k,\pi^{1},\pi^{2}}(s):=E\left[\left.\sum_{i=0}^{+\infty}e^{-i\gamma\Delta}R_{\Delta}^{k}(S_{i},a_{i}^{1},a_{i}^{2})\right\vert S_{0}=s\right],
\]
for $k=1,2$, where $e^{-\gamma\Delta}$ is the discounted factor,
and the running reward $R_{\Delta}^{k}(S_{i},a_{i}^{1},a_{i}^{2})$
is given by
\[
R_{\Delta}^{k}(S_{i},a_{i}^{1},a_{i}^{2})=\underset{\text{profit from the sell order}}{\underbrace{(p_{i}^{a,k}-X_{i}-c)n_{i}^{a,k}}}+\underset{\text{profit from the buy order}}{\underbrace{(X_{i}-p_{i}^{b,k}-c)n_{i}^{b,k}}.}
\]
The constant $c>0$ is the transaction cost. The underlying mechanism
of the above time-discretization is the same with that in Section
\ref{sec:discrete-time HFMM model}.

Same with the continuous-time game $\mathcal{G}_{0}$, every market
maker maximizes the value function of herself, and they choose their
actions simultaneously and independently. We will still focus on the
Nash equilibrium. We begin by introducing some relevant notions. Following
Section 4.1 in \citet{FilarVrieze2012Book_MDP_game}, we consider
the behavior strategies $\Pi_{\Delta}^{k}$ and the stationary strategies
$\Pi_{s}^{k}$ defined below. $\Pi_{\Delta}^{k}$ is defined as the
family of strategies satisfying that, for any $i\geq0$, the policy
$\pi_{i}^{k}\in\mathcal{P}(\mathcal{A})$ taken at time $i$ is measurable
w.r.t. the history $(s_{0},a_{0}^{1},a_{0}^{2},\ldots,s_{i-1},a_{i-1}^{1},a_{i-1}^{2},s_{i})$;
$\Pi_{s}^{k}$ is a subset of $\Pi_{\Delta}^{k}$ satisfying that
the policies are Markov and time invariant, i.e., there exists a mapping
$\pi^{k}:\mathcal{S}\rightarrow\mathcal{P}(\mathcal{A})$ such that
for any $i\geq0$, the policy $\pi_{i}^{k}$ taken at time $i$ is
the randomized policy $\pi^{k}\in\Pi_{\Delta}^{k}$. Here, with a
slight abuse of notations, we use the same notation $\Pi_{s}^{k}$
for the stationary strategies in both continuous-time game $\mathcal{G}_{0}$
and discrete-time game $\mathcal{G}_{\Delta}$ since any stationary
strategy can be characterized by a single mapping $\pi^{k}:\mathcal{S}\rightarrow\mathcal{P}(\mathcal{A})$
in both cases.

For any pair of behavior strategies $(\pi^{1},\pi^{2})=\{(\pi_{i}^{1},\pi_{i}^{2})\}_{i\geq0}$,
we follow the notation convention in Section 4.1 in \citet{FilarVrieze2012Book_MDP_game}
for the reward $R_{\Delta}^{k}(s,\pi_{i}^{1},\pi_{i}^{2})=\sum_{a^{1}\in\mathcal{A},a^{2}\in\mathcal{A}}R_{\Delta}^{k}(s,a^{1},a^{2})\pi_{i}^{1}(a^{1}|s)\pi_{i}^{2}(a^{2}|s)$.
Then, following Definition 4.6.1 in \citet{FilarVrieze2012Book_MDP_game},
we define the Nash equilibrium for the discrete-time stochastic game
$\mathcal{G}_{\Delta}$ as follows.
\begin{defn}
A pair of strategies $(\pi_{\Delta}^{1,*},\pi_{\Delta}^{2,*})\in\Pi_{\Delta}^{1}\times\Pi_{\Delta}^{2}$
is called a Nash equilibrium if
\[
V_{\Delta}^{1,\pi_{\Delta}^{1,*},\pi_{\Delta}^{2,*}}(s)\geq V_{\Delta}^{1,\pi^{1},\pi_{\Delta}^{2,*}}(s)\qquad\text{and}\qquad V_{\Delta}^{2,\pi_{\Delta}^{1,*},\pi_{\Delta}^{2,*}}(s)\geq V_{\Delta}^{2,\pi_{\Delta}^{1,*},\pi^{2}}(s),
\]
for any $(\pi^{1},\pi^{2})\in\Pi_{\Delta}^{1}\times\Pi_{\Delta}^{2}$
and $s\in\mathcal{S}$.

In Theorem \ref{Thm:2player Nash converge} below, we show the convergence
of $\mathcal{G}_{\Delta}$ to $\mathcal{G}_{0}$ in a suitable sense
as $\Delta$ goes to zero.
\end{defn}

\begin{thm}
\label{Thm:2player Nash converge}Under Assumptions \ref{Assump_Q_rate_lambda}
and \ref{Assump_2player_intensity}, there exist pairs of stationary
strategies $(\pi_{0}^{1,*},\pi_{0}^{2,*}),(\pi_{\Delta}^{1,*},\pi_{\Delta}^{2,*})\in\Pi_{s}^{1}\times\Pi_{s}^{2}$
such that $(\pi_{0}^{1,*},\pi_{0}^{2,*})$ (resp. $(\pi_{\Delta}^{1,*},\pi_{\Delta}^{2,*})$)
is a Nash equilibrium in the continuous-time (resp. discrete-time)
market making game $\mathcal{G}_{0}$ (resp. $\mathcal{G}_{\Delta}$).
Moreover, assuming the uniqueness of the Nash equilibrium $(\pi_{0}^{1,*},\pi_{0}^{2,*})$
in the continuous-time game $\mathcal{G}_{0}$, then we have the following
convergence results:

(i) The policies satisfy that $||\pi_{\Delta}^{k,*}(\cdot|s)-\pi_{0}^{k,*}(\cdot|s)||\rightarrow0$
as $\Delta\rightarrow0$ for $k=1,2$, where the norm is defined as
$||\pi(\cdot)-\pi^{\prime}(\cdot)||:=\max_{a\in\mathcal{A}}|\pi(a)-\pi^{\prime}(a)|$
for any $\pi,\pi^{\prime}\in\mathcal{P}(\mathcal{A})$.

(ii) The value functions satisfy that $|V_{\Delta}^{k,\pi_{\Delta}^{1,*},\pi_{\Delta}^{2,*}}(s)-V_{0}^{k,\pi_{0}^{1,*},\pi_{0}^{2,*}}(s)|\rightarrow0$
as $\Delta\rightarrow0$ for any $s\in\mathcal{S}$ and $k=1,2$.
\end{thm}

The assumption on the uniqueness of Nash equilibrium in the continuous-time
game $\mathcal{G}_{0}$ is necessary. Since nonzero-sum games typically
have multiple Nash equilibriums (or no equilibrium), and unlike the
zero-sum game, different equilibriums may have different value functions
(see, e.g., \citet{Abreu1990ECTA_nonzero_sum_games_discrete} and
\citet{sannikov2007ECTA_nonzero_sum_games_continuous}).

\subsection{Nash Q-learning algorithm for solving the equilibrium}

We now introduce the RL algorithm for the learning of Nash equilibrium.
All the learning is conducted in the discrete-time model $\mathcal{G}_{\Delta}$.
We adopt the Nash Q-learning algorithm in \citet{JMLR2003Nash_Q_learning}.
In the stochastic game $\mathcal{G}_{\Delta}$, under the pair of
strategies $(\pi^{1},\pi^{2})$, the Q function for $\text{MM}_{k}$
is defined as follows,
\[
Q_{\Delta}^{k,\pi^{1},\pi^{2}}(s,a^{1},a^{2}):=E[R_{\Delta}^{k}(S_{i},a_{i}^{1},a_{i}^{2})+e^{-\gamma\Delta}V_{\Delta}^{k,\pi^{1},\pi^{2}}(S_{i+1})|S_{i}=s,a_{i}^{1}=a^{1},a_{i}^{2}=a^{2}],
\]
for $k=1,2$, where $R_{\Delta}^{k}$ is the stochastic reward function
of $\text{MM}_{k}$ defined in the previous section, $V_{\Delta}^{k,\pi^{1},\pi^{2}}$
is the value function of $\text{MM}_{k}$, and $e^{-\gamma\Delta}$
is the discounted factor. Indeed, by Theorem 4.6.5 in \citet{FilarVrieze2012Book_MDP_game},
we have that the Nash equilibrium $(\pi_{\Delta}^{1,*},\pi_{\Delta}^{2,*})\in\Pi_{s}^{1}\times\Pi_{s}^{2}$
in the stationary strategy set satisfies that
\[
(\pi_{\Delta}^{1,*}(\cdot|s),\pi_{\Delta}^{2,*}(\cdot|s))=\mathrm{Nash}(Q_{\Delta}^{1,\pi_{\Delta}^{1,*},\pi_{\Delta}^{2,*}}(s,\cdot,\cdot),Q_{\Delta}^{2,\pi_{\Delta}^{1,*},\pi_{\Delta}^{2,*}}(s,\cdot,\cdot)),
\]
for any $s\in\mathcal{S}$. Here, the operator $\mathrm{Nash}()$
is defined as follows: for any two payoff matrices $Q^{1}(a^{1},a^{2})$
and $Q^{2}(a^{1},a^{2})$ for two players, $\mathrm{Nash}(Q^{1}(\cdot,\cdot),Q^{2}(\cdot,\cdot))$
returns a Nash equilibrium of this two-player static game. Thus, using
the above relations between $(\pi_{\Delta}^{1,*},\pi_{\Delta}^{2,*})$
and $(Q_{\Delta}^{1,\pi_{\Delta}^{1,*},\pi_{\Delta}^{2,*}},Q_{\Delta}^{2,\pi_{\Delta}^{1,*},\pi_{\Delta}^{2,*}})$,
we reduce the problem of learning equilibrium to the learning of the
Q functions at equilibrium. In the Nash Q-learning algorithm, the
role of the operator $\mathrm{Nash}()$ is in analogy with the $\arg\max$
operator in the Q-learning algorithm in the single-agent Q-learning
algorithm. Denote by $(Q_{\Delta}^{1,i},Q_{\Delta}^{2,i})$ the pair
of Q functions learned at the $i$th iteration. The details are shown
below.

\begin{algorithm}[h]
\caption{Nash Q-learning algorithm\label{alg:Nash Q-learning}}

\begin{enumerate}
\item Initialize the two Q functions $Q_{\Delta}^{k,0}(s,a^{1},a^{2})$
of $\text{MM}_{k}$ for $k=1,2$.
\item For $i=0,1,2,\ldots,n-1$:
\item $\qquad$Given $S_{i}=s$, with $\varepsilon^{i}(s)$ probability,
sample random actions $(a^{1},a^{2})$ from $\mathcal{A}\times\mathcal{A}$;
Otherwise, sample the actions $(a^{1},a^{2})$ from the strategy pair
$(\pi_{\Delta}^{1,i}(\cdot|s),\pi_{\Delta}^{2,i}(\cdot|s))=\mathrm{Nash}(Q_{\Delta}^{1,i}(s,\cdot,\cdot),Q_{\Delta}^{2,i}(s,\cdot,\cdot))$.
\item $\qquad$Get new data $(s^{\prime},r^{1},r^{2})$ at the time $t=(i+1)\Delta$
given $(S_{i},a_{i}^{1},a_{i}^{2})=(s,a^{1},a^{2})$, where $S_{i+1}=s^{\prime}$
is the new state variable and $r^{k}$ is the reward for $\text{MM}_{k}$.
\item $\qquad$Update the Q functions by 
\[
Q_{\Delta}^{k,i+1}(s,a^{1},a^{2})=r^{k}+\beta^{i}(s,a^{1},a^{2})(Q_{\Delta}^{k,i}(s^{\prime},\hat{a}^{1},\hat{a}^{2})-Q_{\Delta}^{k,i}(s,a^{1},a^{2}))
\]
where $\beta^{i}(s,a^{1},a^{2})$ is the learning rate and the actions
$(\hat{a}^{1},\hat{a}^{2})$ is sampled from the strategy pair $(\pi_{\Delta}^{1,i}(\cdot|s^{\prime}),\pi_{\Delta}^{2,i}(\cdot|s^{\prime}))=\mathrm{Nash}(Q_{\Delta}^{1,i}(s^{\prime},\cdot,\cdot),Q_{\Delta}^{2,i}(s^{\prime},\cdot,\cdot))$.
$=\text{Nash}(Q^{1}(s^{\prime},\cdot,\cdot),Q^{2}(s^{\prime},\cdot,\cdot))$.
\item The learned equilibrium point is given by $(\pi_{\Delta}^{1,n}(\cdot|s),\pi_{\Delta}^{2,n}(\cdot|s))=\mathrm{Nash}(Q_{\Delta}^{1,n}(s,\cdot,\cdot),Q_{\Delta}^{2,n}(s,\cdot,\cdot))$
for $s\in\mathcal{S}$ at the end of the iteration.
\end{enumerate}
\end{algorithm}

Under suitable assumptions, the theoretical guarantee for the convergence
of the algorithm to the true Nash equilibrium is studied in \citet{JMLR2003Nash_Q_learning}.
In Section \ref{sec:Numerical results}, we study numerically the
learned strategies and discuss the choice of learning rate $\beta^{i}$
and exploration probability $\varepsilon^{i}$ in details.

\section{Numerical Studies\label{sec:Numerical results}}

In this section, we conduct extensive experiments to validate our
results\footnote{All the python codes are available at \href{https://github.com/zyh-pku/Reinforcement-Learning-and-Market-Making}{https://github.com/zyh-pku/Reinforcement-Learning-and-Market-Making}}.

\subsection{One-player case}

To validate our theory and demonstrate the effects of sampling frequency,
we conduct the Monte Carlo simulations in the following setup.
\[
\begin{tabular}{|c|c|c|c|c|c|c|c|c|}
\hline  Parameter  &  \ensuremath{N_{P}}  &  \ensuremath{N_{Y}}  &  \ensuremath{\delta_{P}}  &  \ensuremath{\gamma\ }  &  \ensuremath{\alpha\ }  &  \ensuremath{\kappa\ }  &  \ensuremath{\phi\ }  &  \ensuremath{\mathit{c}}\\
\hline  Value  &  2  &  1  &  1/3  &  0.95  &  10.87  &  2  &  0  &  0 
\\\hline \end{tabular}
\]
Here, as defined in the model setup in Sections \ref{sec:continuous-time HFMM model}
and \ref{sec:discrete-time HFMM model}, $N_{P}$ is the number of
different price levels minus one, $N_{Y}$ is the maximum absolute
number of inventory holding, $\delta_{P}$ is the tick size, $\gamma$
is the discounted factor, $\alpha$ and $\kappa$ determine the intensity
function $\lambda(d)=\alpha\exp(-\kappa d)$ of the Poisson process
for the market order flow, $\phi$ determines the penalty function
$\psi(y):=\phi y^{2}$ for inventory holding, and $c$ is the transaction
cost. The Q-matrix of the continuous-time Markov chain for the mid-price
$X_{t}$ is set as
\[
Q_{X}=\left[\begin{array}{ccc}
-5 & 5 & 0\\
\frac{10}{3} & -\frac{20}{3} & \frac{10}{3}\\
0 & 5 & -5
\end{array}\right].
\]
Though our model accommodates the general case where the entries of
$Q_{X}$ are functions of action variable $a$, for simplicity here
we consider a special case where all entries of $Q_{X}$ are constants.
The interpretation is that, the order size of the market maker is
small so that the market impact can be ignored and the dynamics of
mid-price can be assumed to independent with the action variable.
Under this parameter setup, the price space, the mid-price space,
and the inventory space are respectively given by
\[
\mathcal{S}_{P}:=\{0,\delta_{P},2\delta_{P}\},\text{ }\mathcal{S}_{X}:=\left\{ \frac{1}{2}\delta_{P},\delta_{P},\frac{3}{2}\delta_{P}\right\} ,\text{ and }\mathcal{S}_{Y}:=\{-1,0,1\}.
\]
To give an example of how the above parameter setup for price dynamics
works, we consider the case that the time increment is $\Delta=0.1$.
Then, for the discretized model $\mathcal{M}_{\Delta}$, the transition
probability matrix of mid-price $X_{i}$ is given by $P_{X}(\Delta|a)=I_{|\mathcal{S}_{X}|}-Q_{X}\Delta$.
By calculation we have that, when $X_{i}=\frac{1}{2}\delta_{P}$ (resp.
$\frac{3}{2}\delta_{P}$), $X_{i+1}$ goes up (resp. down) to $\delta_{P}$
with probability $1/2$ and stay the same otherwise; when $X_{i}=\delta_{P}$,
the probabilities that $X_{i+1}$ goes up, goes down, or stay the
same are all equal to $1/3$.

The above parameters and the value of $\Delta$ determine the discrete-time
MDP $\mathcal{M}_{\Delta}$. To validate our theory on the effect
of $\Delta$, we set $\Delta$ according to the decreasing sequence
$\{10^{-1-2k/9}\}_{k=0,1,...,9}$. The optimal value functions $V_{\Delta}^{*}(s)$
and the optimal policies $\pi_{\Delta}^{*}(s)$ of the discrete-time
MDP $\mathcal{M}_{\Delta}$ under different values of $\Delta$ can
be computed using the Bellman equations. We find that for all the
different values of $\Delta$ in our experiment, the optimal policy
$\pi_{\Delta}^{*}(s)$ for the discrete-time MDP $\mathcal{M}_{\Delta}$
are all identical to $\pi^{*}(s)=(p_{a}^{*}(s),p_{b}^{*}(s))$ given
as follows. For the state variable $s=(x,y)\in\mathcal{S}_{X}\times\mathcal{S}_{Y}$,
we have that $p_{a}^{*}(x,y)=2\delta_{P}$ for any $x\in\mathcal{S}_{X}$
and $y=0,1$, $p_{b}^{*}(x,y)=0$ for any $x\in\mathcal{S}_{X}$ and
$y=0,-1$; when $y=-1$ (resp. $1$), ask (resp. buy) order is banned
and thus $p_{a}^{*}(x)$ (resp. $p_{b}^{*}(x)$) is no need to set.
The policy $\pi^{*}(s)$ and the value function $V_{\Delta}^{*}(s)$
with the smallest $\Delta$ in our experiment is an approximate solution
to the optimality equations (\ref{HJB for Value continuous}) that
characterize the continuous-time MDP $\mathcal{M}_{0}$. More precisely,
the numerical results show that $\max_{s\in\mathcal{S}}|\gamma V(s)-\max_{a\in\mathcal{A}}(f(s,a)+\sum_{s^{\prime}\in\mathcal{S},s^{\prime}\neq s}\lambda_{s^{\prime}}(s,a)V(s^{\prime})-\lambda(s,a)V(s))|=0.008$,
where the difference is small enough and the convergence of optimal
and value function in $\mathcal{M}_{\Delta}$ to those in $\mathcal{M}_{0}$
is validated. Since $\pi^{*}(s)$ is a pure strategy, we obtain that
the optimal policy of the continuous-time MDP $\mathcal{M}_{0}$ is
given by $\pi_{0}^{*}(s)=\pi^{*}(s)$, and the optimal value function
$V_{0}^{*}(s)$ is approximated by $V_{\Delta}^{*}(s)$ with the smallest
$\Delta$ in our experiment.

In Figure \ref{figure:V converge vs Delta}, we plot the optimal value
function $V_{\Delta}^{*}(s)$ against $\Delta$. The nine subplots
correspond to $V_{\Delta}^{*}(s)$ since the state space $\mathcal{S}=\mathcal{S}_{X}\times\mathcal{S}_{Y}$
is a $3\times3$-dimensional finite space. In each subplot, the dashed
line represents the optimal value function $V_{0}^{*}(s)$ of the
continuous-time MDP $\mathcal{M}_{0}$ at this state. The dashed line
is the value function $V_{\Delta}^{*}(s)$ under the smallest $\Delta$,
which is an approximation of $V_{0}^{*}(s)$ as discussed previously.
As $\Delta$ goes to zero, the convergence of $V_{\Delta}^{*}(s)$
can be clearly seen from these subplots, since the lines of $V_{\Delta}^{*}(s)$
get closer and closer to the level of $V_{0}^{*}(s)$ as $\Delta$
becomes smaller and smaller. These observations, together with that
for the optimal policies mentioned above, demonstrate our results
in Theorem \ref{Thm:value discretize error} on the convergence of
the discrete-time MDP $\mathcal{M}_{\Delta}$ to the continuous-time
MDP $\mathcal{M}_{0}$ in terms of the value functions and the policies.

\begin{figure}[p]
\centering%
\begin{tabular}{c}
\includegraphics[scale=0.58]{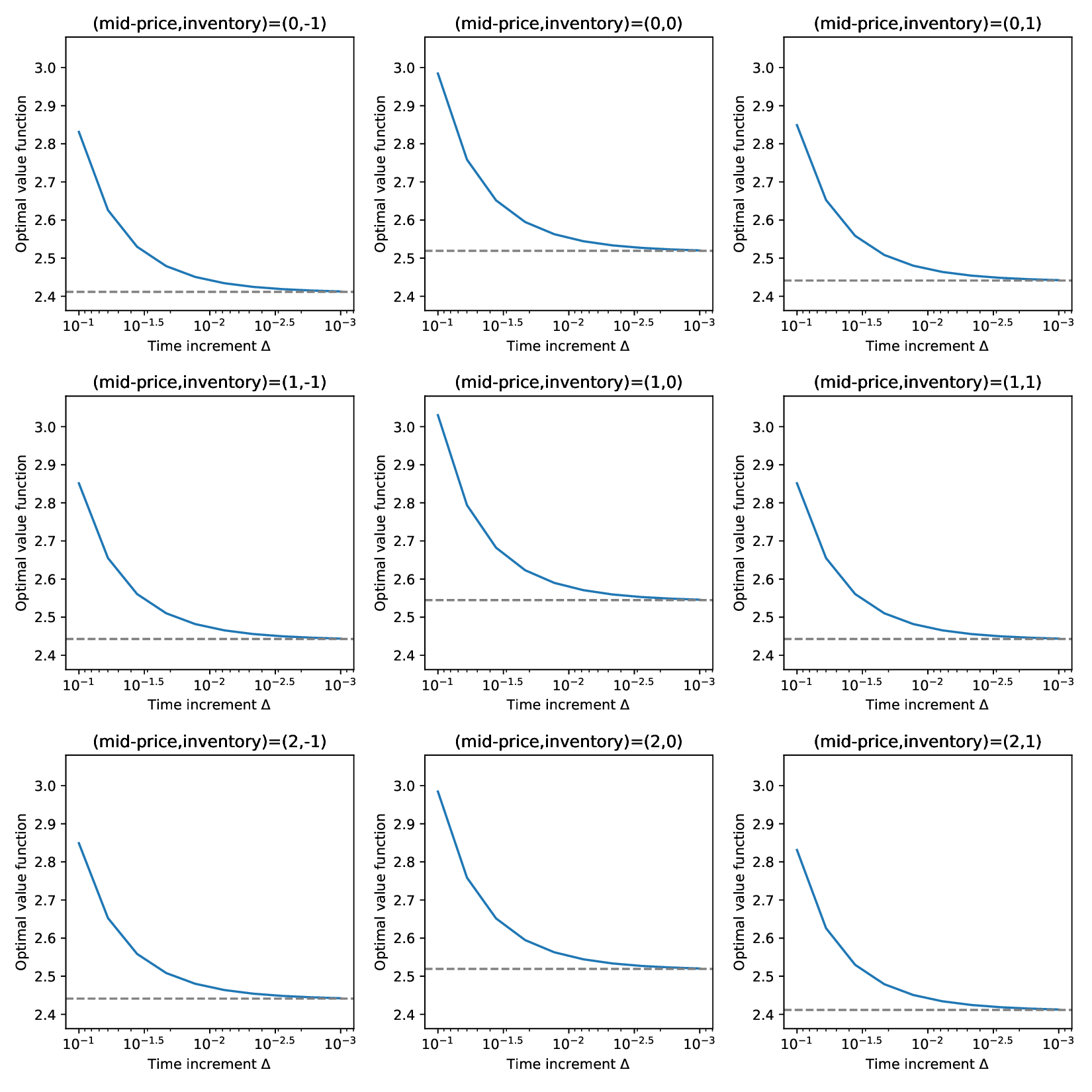}\tabularnewline
\end{tabular}

\caption{The convergence of optimal value function $V_{\Delta}^{*}(s)$ against
$\Delta$\label{figure:V converge vs Delta}}
\end{figure}

Then we shift the gear to the aspect of the RL method in our theory.
For each of the discrete-time MDP $\mathcal{M}_{\Delta}$, we run
Q-learning algorithm with $\varepsilon$-greedy exploration as described
in Section \ref{sec:Q-learning sample cplx}. Here, the sequence $\{\varepsilon^{(n)}(s)\}$
of the exploration probabilities is set as $\varepsilon^{(n)}(s)=\max(\varepsilon_{0},\rho_{0}\rho^{\left\lfloor N(s,n)/M\right\rfloor })$,
where $\varepsilon_{0}$ is the smallest value we permit, $M$ is
called the learning epoch, $\rho\in(0,1)$ is the decaying rate of
the exploration probability, and $N(s,n)$ is the number of times,
until the $n$th iteration, that we visited the state $s$. Under
different values of $\Delta$, the hyperparameters $\rho_{0}$, $\rho$,
and $M$ vary according to our fine tuning, while the parameter $\varepsilon_{0}$
is fixed at $\varepsilon_{0}=10^{-5}$. For the learning rate $\beta^{(n)}(s,a)=(N(s,a,n))^{-\omega}$
introduced in Section \ref{sec:Q-learning sample cplx}, the hyperparameter
$\omega$ we use are close to $0.5$ for all different values of $\Delta$:
for $\{\Delta_{k}=10^{-1-2k/9}\}_{k=0,1,...,9}$, we use $\omega=0.501$
when $\Delta=$$\Delta_{2}$, $\Delta_{3}$, and $\Delta_{6}$, and
we use $\omega=0.5001$ for all other values of $\Delta$.

To capture the sample complexity using numerical simulations, we record
the iteration number $N_{\Delta}$, which is the smallest iteration
steps $n$ such that $||V_{\Delta}^{(n)}(\cdot)-V_{\Delta}^{*}(\cdot)||\leq0.1$,
when applying the Q-learning algorithm to the discrete-time MDP for
each value of $\Delta$.

\begin{figure}[htbp]
\centering%
\begin{tabular}{c}
\includegraphics[scale=0.69]{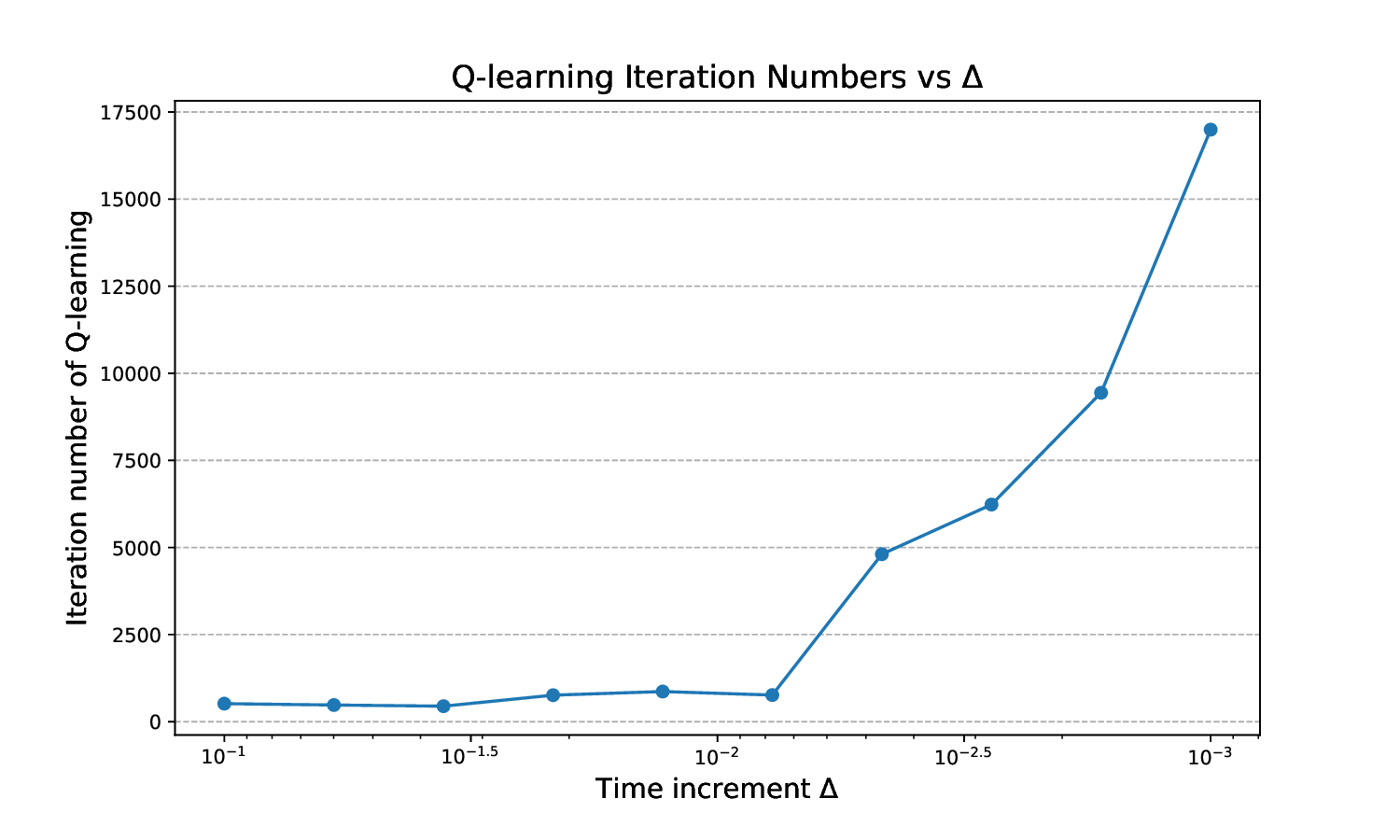}\tabularnewline
\end{tabular}
\caption{The trend of sample complexity against $\Delta$\label{figure:sample cplx vs Delta}}
\end{figure}

In Figure \ref{figure:sample cplx vs Delta}, we plot the iteration
number $N_{\Delta}$ against $\Delta$, where the x-axis is in logarithmic
scale. This increasing trend of the iteration numbers demonstrate
our theory about the sample complexity, i.e., the sample complexity
upper bound will increase when $\Delta$ decreases. In particular,
though the upper bound in Theorem \ref{Thm:sample cplx bound Delta}
depends on the hyperparameter $\omega$, the values of $\omega$ we
use are all the same regardless the difference smaller than $10^{-2}$.
So, the only difference in the upper bound is the values of $\Delta$
and the upward trend in Figure \ref{figure:sample cplx vs Delta}
does validate Theorem \ref{Thm:sample cplx bound Delta}.

\subsection{Two-player case}

First, we examine the convergence of Nash equilibrium. We start with
the following setup. For the intensity functions $\Gamma^{a,k}$ and
$\Gamma^{b,k}$, their building blocks $\Upsilon^{-}$ and $\Upsilon^{+}$are
set to
\[
\Upsilon^{-}(d)=\alpha\exp(-\kappa d)\text{ and }\Upsilon^{+}(d)=\frac{1}{2}\sqrt{1+3\exp(-\kappa d)}.
\]
This choice of intensity functions connects to the single-agent case
in the following sense: when the two market makers always quote the
same prices, i.e., $p_{t}^{1,a}=p_{t}^{2,a}=p_{t}^{a}$ and $p_{t}^{1,b}=p_{t}^{2,b}=p_{t}^{b}$,
the intensity of market flows they will receive is exactly same with
the single-agent case because by calculations we have that
\[
\Gamma^{a,k}(x,p^{a},p^{a})=\frac{\Upsilon^{-}(|p^{a}-x|)}{\Upsilon^{+}(0)}=\lambda(|p^{a}-x|)\text{ and }\Gamma^{b,k}(x,p^{b},p^{b})=\frac{\Upsilon^{-}(|p^{b}-x|)}{\Upsilon^{+}(0)}=\lambda(|p^{b}-x|).
\]
We set the model parameters the same as the setup in the single-agent
case.

The above parameters and the value of $\Delta$ determine the discrete-time
game $\mathcal{G}_{\Delta}$. To validate our theory on the effect
of $\Delta$, we set $\Delta$ according to the decreasing sequence
$\{10^{-1-2k/9}\}_{k=0,1,...,9}$. For the true Nash equilibrium,
we solve it through value iteration following the same manner of Bellman
iteration in the single-agent case. In this example, the model is
symmetric w.r.t. $\text{MM}_{1}$ and $\text{MM}_{2}$, and thus the
strategies and the value functions of $\text{MM}_{1}$ and $\text{MM}_{2}$
should be the same with each other, i.e., $\pi_{\Delta}^{1,*}=\pi_{\Delta}^{2,*}$
and $V_{\Delta}^{1,\pi_{\Delta}^{1,*},\pi_{\Delta}^{2,*}}(s)=V_{\Delta}^{2,\pi_{\Delta}^{1,*},\pi_{\Delta}^{2,*}}(s)$.
Also, under our current setup, we find from the numerical results
that there is exactly one Nash equilibrium for $\mathcal{G}_{\Delta}$
under different values of $\Delta$.

The convergence of Nash equilibrium is validated by the numerical
results as follows. First, we find that all the Nash equilibriums
computed in $\mathcal{G}_{\Delta}$ under different values of $\Delta$
are exactly the same with each other. Second, the convergence of the
value function $V_{\Delta}^{k,\pi_{\Delta}^{1,*},\pi_{\Delta}^{2,*}}(s)$
at Nash equilibrium is illustrated in Figure \ref{figure:Nash value vs Delta}.
These equilibrium points under $\mathcal{G}_{\Delta}$ are all equal
to a pure strategy $\pi_{0}^{*}(s)$, i.e., $\pi_{\Delta}^{1,*}(s)=\pi_{\Delta}^{2,*}(s)=\pi_{0}^{*}(s)=(p_{a}^{*}(s),p_{b}^{*}(s))$,
where $p_{a}^{*}(\frac{1}{2}\delta_{P})=p_{a}^{*}(\delta_{P})=p_{a}^{*}(\frac{3}{2}\delta_{P})=2\delta_{P}$
and $p_{b}^{*}(\frac{1}{2}\delta_{P})=p_{b}^{*}(\delta_{P})=p_{b}^{*}(\frac{3}{2}\delta_{P})=0$.
The strategy is the reaction of the market maker to the state variable
$s$ which is the mid-price in the state space $\mathcal{S}_{X}=\{\frac{1}{2}\delta_{P},\delta_{P},\frac{3}{2}\delta_{P}\}$.
The interpretation of the pure strategy $\pi_{0}^{*}(s)=(p_{a}^{*}(s),p_{b}^{*}(s))$
is that, no matter what level of mid-price is, both $\text{MM}_{1}$
and $\text{MM}_{2}$ will quote at the highest ask price $2\delta_{P}$
(resp. lowest bid price $0$) to maximize their expected profit. In
particular, this special optimal strategy is due to the simple setup
we are currently studying, and the optimal strategy should be much
more complicated and variant with different states in general setup,
e.g., the dimension of state space $\mathcal{S}_{X}$ is high.

Besides the identity among $(\pi_{\Delta}^{1,*}(s),\pi_{\Delta}^{2,*}(s))$
under different $\Delta$ and the convergence trend of $(V_{\Delta}^{1,\pi_{\Delta}^{1,*},\pi_{\Delta}^{2,*}}(s),V_{\Delta}^{2,\pi_{\Delta}^{1,*},\pi_{\Delta}^{2,*}}(s))$
shown in Figure \ref{figure:Nash value vs Delta}, the convergence
results in Theorem \ref{Thm:2player Nash converge} are further validated
by the findings below. We find that, the strategy pair $(\pi_{0}^{1}(s),\pi_{0}^{2}(s))=(\pi_{0}^{*}(s),\pi_{0}^{*}(s))$
and the value function pair $(V_{0}^{1}(s),V_{0}^{2}(s))=(V_{\Delta}^{1,\pi_{\Delta}^{1,*},\pi_{\Delta}^{2,*}}(s),V_{\Delta}^{2,\pi_{\Delta}^{1,*},\pi_{\Delta}^{2,*}}(s))$
with the smallest $\Delta$ in our experiment is an approximate solution
to the optimality equations (\ref{cont V1 game Bellman equation})--(\ref{cont V2 game Bellman equation})
that characterize the Nash equilibrium in continuous-time game $\mathcal{G}_{0}$.
More precisely, the numerical results show that $\max_{s\in\mathcal{S}}|\gamma V_{0}^{1}(s)-\max_{a\in\mathcal{A}}\bigl(r^{1}(s,a,\pi_{0}^{*})+\sum_{s^{\prime}\in\mathcal{S}}\lambda(s^{\prime}|s,a,\pi_{0}^{*})V_{0}^{1}(s^{\prime})\bigr)|=0.057$
and $\max_{s\in\mathcal{S}}|\gamma V_{0}^{2}(s)-\max_{a\in\mathcal{A}}\bigl(r^{2}(s,\pi_{0}^{*},a)+\sum_{s^{\prime}\in\mathcal{S}}\lambda(s^{\prime}|s,\pi_{0}^{*},a)V_{0}^{2}(s^{\prime})\bigr)|=0.057$,
where the difference is small enough and the convergence of Nash equilibrium
strategy and value function in $\mathcal{G}_{\Delta}$ to those in
$\mathcal{G}_{0}$ is validated.

\begin{figure}[htbp]
\centering%
\begin{tabular}{c}
\includegraphics[scale=0.69]{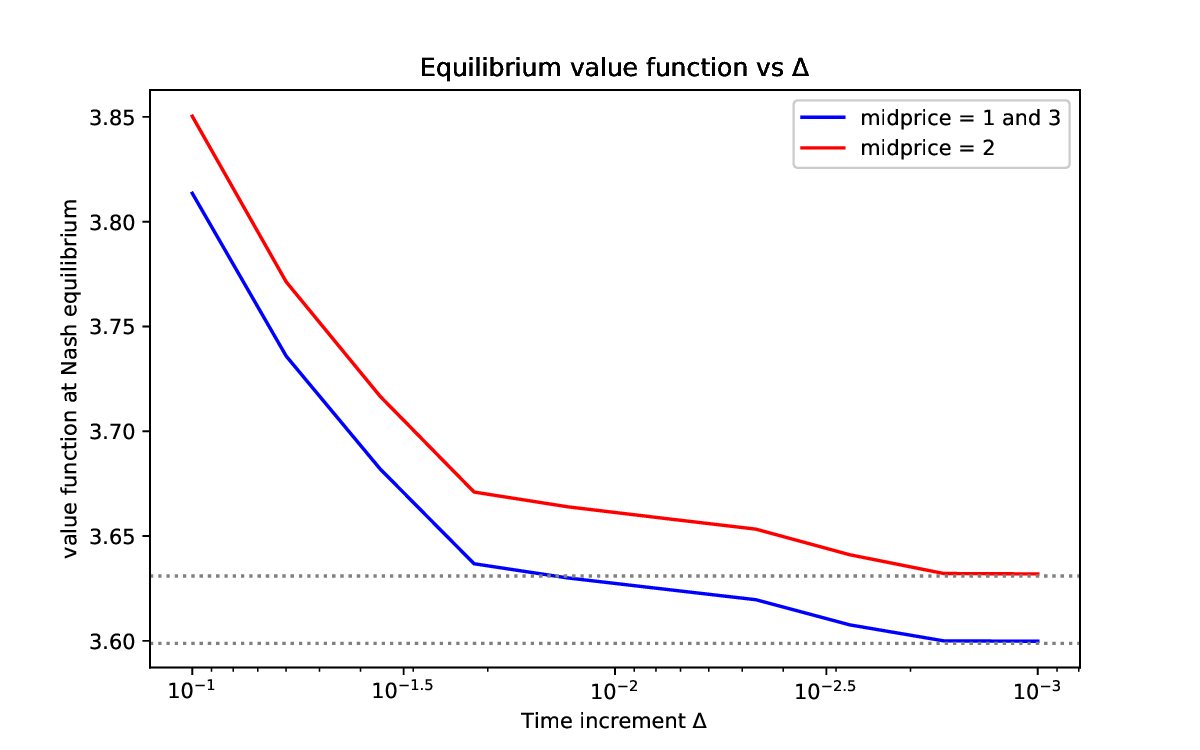}\tabularnewline
\end{tabular}

\caption{The convergence of equilibrium value function $V_{\Delta}^{k,\pi_{\Delta}^{1,*},\pi_{\Delta}^{2,*}}(s)$
against $\Delta$\label{figure:Nash value vs Delta}}
\end{figure}

Then, we shift the gear to the aspect of the RL method in our theory.
Under the discrete-time game $\mathcal{G}_{\Delta}$ with time increment
$\Delta=0.1$, we run Nash Q-learning algorithm with $\varepsilon$-greedy
exploration as described in Algorithm \ref{alg:Nash Q-learning} in
Section \ref{sec:Q-learning sample cplx}. Here, the sequence $\{\varepsilon^{n}(s)\}$
of the exploration probabilities is set as $\varepsilon^{n}(s)=\max(\varepsilon_{0},\rho_{0}\rho^{\left\lfloor N(s,n)/M\right\rfloor })$,
where $\varepsilon_{0}$ is the smallest value we permit, $M$ is
called the learning epoch, $\rho\in(0,1)$ is the decaying rate of
the exploration probability, and $N(s,n)$ is the number of times,
until the $n$th iteration, that we visited the state $s$. The learning
rate sequence $\{\beta^{n}(s,a^{1},a^{2})\}$ is similarly set as
$\beta^{n}(s,a^{1},a^{2})=\eta_{0}\eta^{\left\lfloor N(s,a^{1},a^{2},n)/M_{b}\right\rfloor }$,
where $N(s,a^{1},a^{2},n)$ is the number of times, until the $n$th
iteration, that we visited the state action pair $(s,a^{1},a^{2})$.

In our experiment, the hyperparameters $\varepsilon_{0}$, $\rho_{0}$,
$\rho$, $M$, $\eta_{0}$, $\eta$, and $M_{b}$ are tuned to guarantee
the convergence of the Nash Q-learning algorithm. We plot in Figure
\ref{figure:value learning error}
how the learning errors of value functions and policies decay w.r.t.
the iteration steps of the Nash Q-learning algorithm. The value function
error and the policy error are given by $\max_{s\in\mathcal{S}}|V_{\Delta}^{k,i}(s)-V_{\Delta}^{k,\pi_{\Delta}^{1,*},\pi_{\Delta}^{2,*}}(s)|$
and $\max_{s\in\mathcal{S}}||\pi_{\Delta}^{k,*}(\cdot|s)-\pi_{0}^{k,*}(\cdot|s)||$,
respectively, where $||\pi(\cdot)-\pi^{\prime}(\cdot)||=\max_{a\in\mathcal{A}}|\pi(a)-\pi^{\prime}(a)|$
for any $\pi,\pi^{\prime}\in\mathcal{P}(\mathcal{A})$. The learned
strategy is exactly same with the true strategy, i.e., the policy
error is zero, and the value function error ends with $(0.09,0.05)$,
demonstrating the effectiveness of the Nash Q-learning algorithm.
In Figure \ref{figure:value learning error},
we use red color to mark the steps that have nonzero policy error.
We find that, the policy error often follows from a rise in the value
function error, e.g., around 400th and 800th steps in Figure \ref{figure:value learning error} (left);
when the rise in value function error is small, the RL algorithm can
still perfectly learn the true policy, e.g., around 800th steps in
Figure \ref{figure:value learning error} (right).

\begin{figure}[htbp]
\centering%
\begin{tabular}{c}
\includegraphics[scale=0.42]{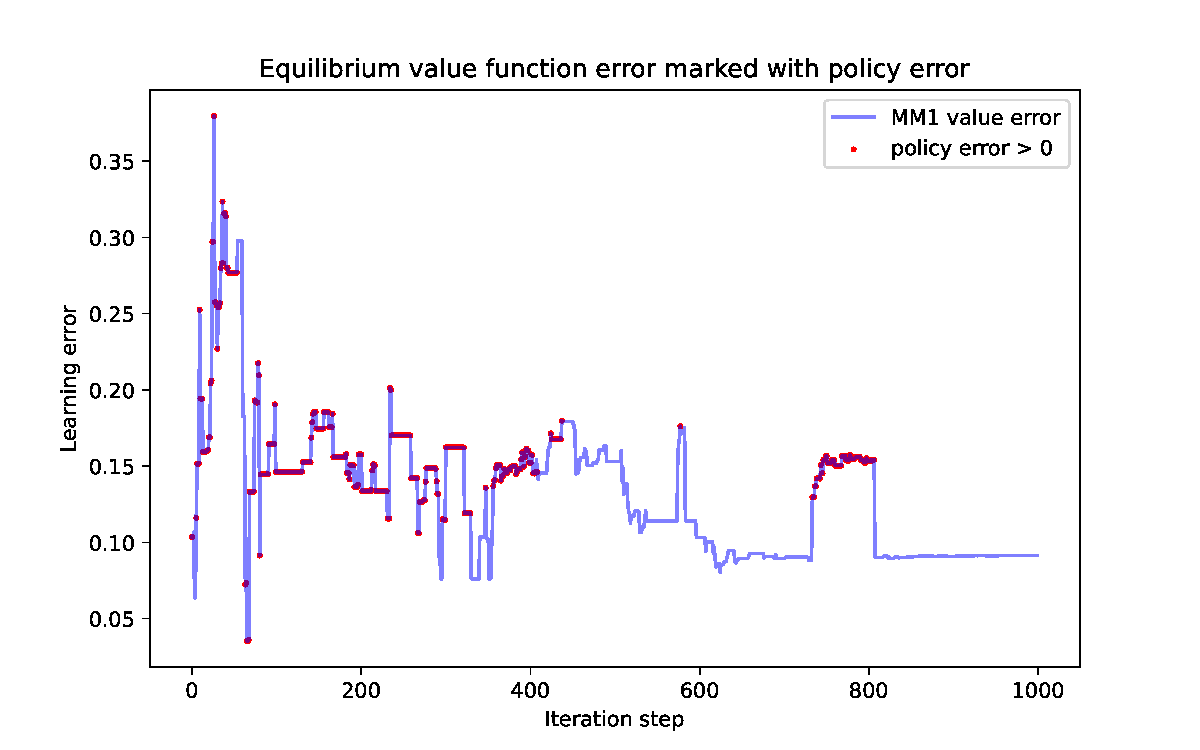}
\includegraphics[scale=0.42]{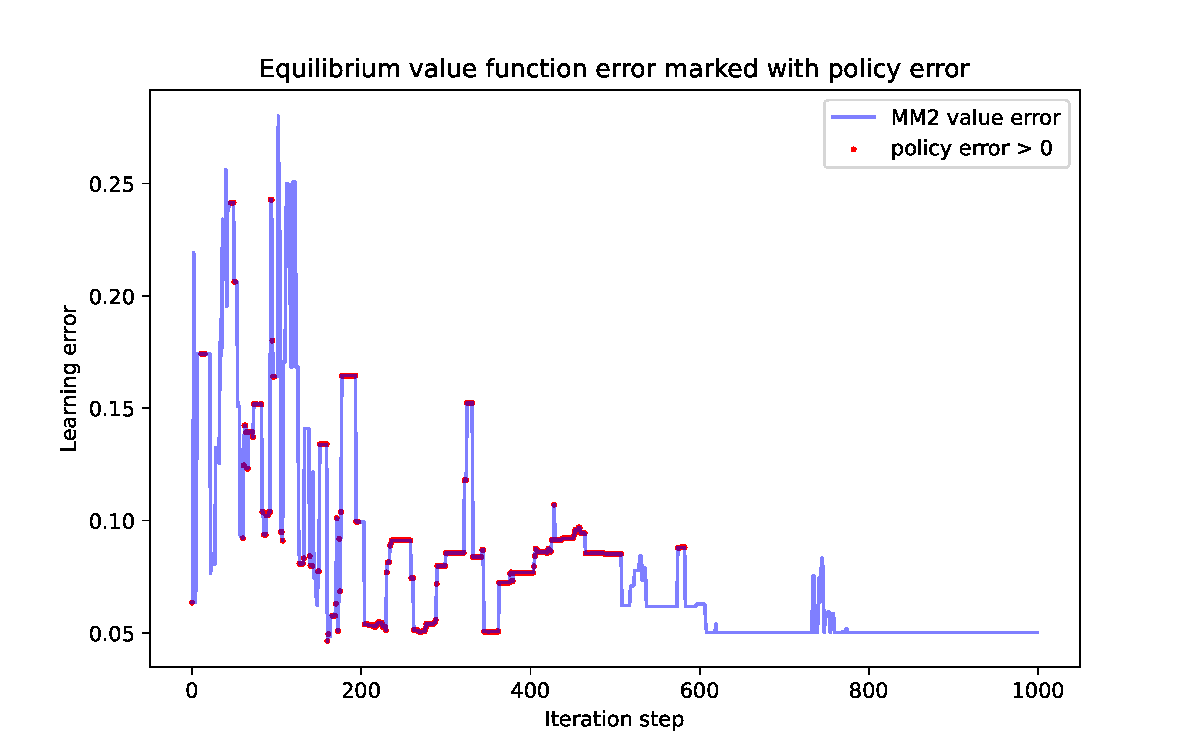}
\end{tabular}
\caption{$\text{MM}_{1}$ (left) and $\text{MM}_{2}$ (right) equilibrium value function learning error\label{figure:value learning error}}
\end{figure}



\section{Conclusion and Discussion\label{sec:Conclusion}}

In this paper, we provide a novel and comprehensive theoretical analysis
for RL algorithm in high-frequency market making problem. We target
the effects of sampling frequency on the RL method, and find an interesting
tradeoff between the accuracy and the complexity of RL algorithm when
changing the frequency. Both of these two metrics are of realistic
importance because, the error captures the accuracy of our estimation
of the expected profit and the complexity measures the total transaction
costs in some way as every iteration of RL algorithm corresponds to
a quote of market maker. Furthermore, we extend our model to study
the price competitions for multiple market makers under the game theoretical
framework. We establish the convergence of Nash equilibrium from discretized
model to continuous-time model, and apply the Nash Q-learning to solve
the equilibrium. Our theory is applicable to any discretized continuous-time
Markov decision process, and thus can be used to study many other
high-frequency financial decision making problems, e.g., optimal executions.
Our results also provide a guide for practitioners on the choice of
sampling frequency, which should depend on what aspect of the algorithm
is primarily concerned by them.

While the current scope of our theory is considerably wide-ranging,
it is possible to further widen it and there are lots of topics that
are worth pursuing. For instance, how to generalize our theory for
more complex market making models and more sophisticated RL methods,
e.g., deep RL algorithm? Is it possible to sharpen the sample complexity
upper bound or even make it tight, such that we have a better characeterization
for the effects of sampling frequency? What will happen to our theoretical
analysis if we switch to the multi-asset case where multiple assets
can be traded in the market? Among many directions, these topics can
be investigated in future research.

\appendix
\bibliographystyle{apalike}
\bibliography{newadded_RL_MM,mainbib}

\section{Proof of Theorem \ref{Thm:value discretize error}}

\subsection{General results for controlled Markov chains\label{sec:general MDP convergence results}}

To prove our convergence results, we consider a more general case
for continuous-time Markov chains as follows, which is potentially
useful for future research.

For the continuous-time MDP, the state variable $S_{t}$ is a continuous-time
Markov chain with finite state space $\mathcal{S}$. For any $s\in\mathcal{S}$,
the rate parameters are given by $\lambda(s,a)$ and $\{\lambda_{s^{\prime}}(s,a)\}_{s^{\prime}\in\mathcal{S},s^{\prime}\neq s}$
satisfying that $\sum_{s^{\prime}\in\mathcal{S},s^{\prime}\neq s}\lambda_{s^{\prime}}(s,a)=\lambda(s,a)$
and $\lambda_{s^{\prime}}(s,a)\geq0$. More precisely, when $S_{t}=s$
and $a_{t}=a$ at time $t$, then after a random waiting time $\tau$
which follows an exponential distribution with parameter $\lambda(s,a)$,
there will be a state transition in $S_{t}$ and the probability is
given by $P(S_{t+\tau}=s^{\prime})=\lambda_{s^{\prime}}(s,a)/\lambda(s,a)$.
We denote by $S_{t}^{\pi}$ the state variable under the policy $\pi$.
The reward function is given by
\begin{equation}
V_{0}^{\pi}(s):=E\left[\left.\int_{0}^{+\infty}e^{-\gamma t}f(S_{t}^{\pi},a_{t})dt\right\vert S_{0}=s\right],\label{value function def cont}
\end{equation}
where the action $a_{t}$ is given by the policy $\pi$.

For the discrete-time MDP $\mathcal{M}_{\Delta}$, the state variable
$S_{i}$ is a discrete-time Markov chain. When the agent take the
action $a$ at time $i$, the transition probabilities are given by
$P(S_{i+1}=s)=p_{\Delta}(s|s,a)=1-\lambda(s,a)\Delta+O(\Delta^{2})$,
and $P(S_{i+1}=s^{\prime})=p_{\Delta}(s^{\prime}|s,a)=\lambda_{s^{\prime}}(s,a)\Delta+O(\Delta^{2})$
for any $s^{\prime}\in\mathcal{S}$ and $s^{\prime}\neq s$, where
we assume that $\Delta$ is sufficiently small so that all these probabilities
are in $[0,1]$. We denote by $S_{i}^{\pi}$ the state variable under
the policy $\pi$. The reward function is given by
\[
V_{\Delta}^{\pi}(s):=E\left[\left.\sum_{i=0}^{+\infty}e^{-i\gamma\Delta}f(S_{i}^{\pi},a_{i})\Delta\right\vert S_{0}=s\right],
\]
where the action $a_{t}$ is given by the policy $\pi$ and the function
$f$ is the same with that in (\ref{value function def cont}).

We assume that the functions $f(s,a)$ and $\lambda(s,a)$ are uniformly
bounded over $\mathcal{S}\times\mathcal{A}$. Then since $\mathcal{S}\times\mathcal{A}$
is finite, automatically we have that these functions are Lipschitz.
Thus, there exist constants $C_{B}^{f}$ and $C_{Lip}^{f}$ such that
$|f(s,a)|<C_{B}^{f}$ and $|f(s,a)-f(s^{\prime},a^{\prime})|<C_{Lip}^{f}(|s-s^{\prime}|+|a-a^{\prime}|)$
for any $s,s^{\prime}\in\mathcal{S}$ and $a,a^{\prime}\in\mathcal{A}$.
The same holds for $\lambda(s,a)$. The optimal value functions $V_{0}^{*}(s)$
for $\mathcal{M}_{0}$ and $V_{\Delta}^{*}(s)$ for $\mathcal{M}_{\Delta}$
are defined in the same manner as we did in Section \ref{sec:discrete-time HFMM and converge}.
Then we have that
\begin{lem}
\label{Lemma:Bellman equations for two MDPs}For the continuous-time
MDP $\mathcal{M}_{0}$, the optimal value function $V_{0}^{*}(\cdot)$
is the unique solution to the dynamic programming optimality equation
as follows
\begin{equation}
\gamma V(s)=\max_{a\in\mathcal{A}}(f(s,a)+\sum_{s^{\prime}\in\mathcal{S},s^{\prime}\neq s}\lambda_{s^{\prime}}(s,a)V(s^{\prime})-\lambda(s,a)V(s)),\label{HJB for Value continuous}
\end{equation}
There exists a stationary Markov policy $\pi_{0}^{*}(\cdot)$, such
that $V_{0}^{*}(s)=V_{0}^{\pi_{0}^{*}}(s)$.

For the discrete-time MDP $\mathcal{M}_{\Delta}$, the optimal value
function $V_{\Delta}^{*}(s)$ is the unique solution to the Bellman
equation as follows
\begin{equation}
V(s)=\max_{a\in\mathcal{A}}(f(s,a)\Delta+e^{-\gamma\Delta}\sum_{s^{\prime}\in\mathcal{S}}p_{\Delta}(s^{\prime}|s,a)V(s)).\label{Bellman for Value discrete}
\end{equation}
There exists a stationary Markov policy $\pi_{\Delta}^{*}(\cdot)$,
such that $V_{\Delta}^{*}(s)=V_{\Delta}^{\pi_{\Delta}^{*}}(s)$.
\end{lem}

\begin{proof}
The continuous-time results follow from \citet{kaku1971CTMDP_Bellman}.
The discrete-time results follow from Theorem 1.13 in \citet{gihSkorohod2012controlled}.
\end{proof}
Then we are ready to prove the main results on the convergence.
\begin{lem}
\label{Lemma:value discretize error general}Assuming the uniqueness
of the optimal policies in $\mathcal{M}_{0}$ and $\mathcal{M}_{\Delta}$,
then we have that, there exists $\Delta_{0}$ such that, for any $\Delta\in(0,\Delta_{0})$,
it holds: (i) The policies $\pi_{\Delta}^{*}(\cdot)$ and $\pi^{*}(\cdot)$
are the identical, i.e., $\pi_{\Delta}^{*}(s)=\pi^{*}(s)$ for any
$s\in\mathcal{S}$. (ii) The bound for the optimal value functions
is given by
\[
||V_{\Delta}^{*}(\cdot)-V_{0}^{*}(\cdot)||\leq C_{V}\Delta,
\]
where the norm is defined as $||V_{\Delta}^{*}(\cdot)-V_{0}^{*}(\cdot)||:=\max_{s\in\mathcal{S}}|V_{\Delta}^{*}(s)-V_{0}^{*}(s)|$,
and $C_{V}>0$ is a constant that does not depend on $\Delta$.
\end{lem}

The proof consists of three steps.

Step 1 --- Prove the convergence of value function as $\Delta\rightarrow0$,
i.e., $\lim_{\Delta\rightarrow0}V_{\Delta}^{*}(\cdot)=V_{0}^{*}(\cdot)$

Since the state space $\mathcal{S}$ is finite, we can regard the
value function $V_{\Delta}^{\pi}(\cdot)$ as a vector in $\mathbb{R}^{|\mathcal{S}|}$.
We obtain by the boundness of $f$ that $|V_{\Delta}^{\pi}(s)|<\sum_{i=0}^{+\infty}e^{-i\gamma\Delta}C_{B}^{f}\Delta=C_{B}^{f}\frac{\Delta}{1-e^{-\gamma\Delta}}<2C_{B}^{f}\frac{1}{\gamma}$,
which holds when $\Delta$ is sufficiently small. So there exists
$\Delta_{1}>0$ such that, for any $\Delta\in(0,\Delta_{1})$, $V_{\Delta}^{\pi}(\cdot)$
is always in a compact subspace of $\mathbb{R}^{|\mathcal{S}|}$.
Then for any sequence of $\Delta$ in $(0,\Delta_{1})$ that converges
to zero, there always exists a subsequence, which we denote as $\{\Delta_{k}\}$,
such that $\lim_{k\rightarrow\infty}V_{\Delta_{k}}^{*}(\cdot)=V^{*}(\cdot)$.
To prove $\lim_{\Delta\rightarrow0}V_{\Delta}^{*}(\cdot)=V_{0}^{*}(\cdot)$,
it suffices to prove that the subsequence limit $V^{*}(\cdot)$ is
always equal to $V_{0}^{*}(\cdot)$.

To prove that $V^{*}(s)=V_{0}^{*}(s)$ for any $s\in\mathcal{S}$,
we will use the Bellman equations in Lemma \ref{Lemma:Bellman equations for two MDPs}.
Using the formulas of $p_{\Delta}(s^{\prime}|s,a)$, we rewrite the
Bellman equation (\ref{Bellman for Value discrete}) as
\[
V_{\Delta}^{*}(s)=\max_{a\in\mathcal{A}}\left(f(s,a)\Delta+e^{-\gamma\Delta}\sum_{s^{\prime}\in\mathcal{S},s^{\prime}\neq s}\lambda_{s^{\prime}}(s,a)V_{\Delta}^{*}(s^{\prime})\Delta-e^{-\gamma\Delta}\lambda(s,a)V_{\Delta}^{*}(s)\Delta+e^{-\gamma\Delta}V_{\Delta}^{*}(s)+R_{1}(\Delta|s,a)\right),
\]
where $|R_{1}(\Delta|s,a)|<C_{R_{1}}\Delta^{2}$ and the constant
$C_{R_{1}}$ does not depend on $(\Delta,s,a)$. So we obtain that
\[
\frac{1}{\Delta}(1-e^{-\gamma\Delta})V_{\Delta}^{*}(s)=\max_{a\in\mathcal{A}}\left(f(s,a)+e^{-\gamma\Delta}\sum_{s^{\prime}\in\mathcal{S},s^{\prime}\neq s}\lambda_{s^{\prime}}(s,a)V_{\Delta}^{*}(s^{\prime})-e^{-\gamma\Delta}\lambda(s,a)V_{\Delta}^{*}(s)+\frac{1}{\Delta}R_{1}(\Delta|s,a)\right),
\]
and thus by the Taylor expansion of $e^{-\gamma\Delta}$, we have
that
\[
\gamma V_{\Delta}^{*}(s)=\max_{a\in\mathcal{A}}\left(f(s,a)+\sum_{s^{\prime}\in\mathcal{S},s^{\prime}\neq s}\lambda_{s^{\prime}}(s,a)V_{\Delta}^{*}(s^{\prime})-\lambda(s,a)V_{\Delta}^{*}(s)+R_{2}(\Delta|s,a)\right),
\]
where $|R_{2}(\Delta|s,a)|<C_{R_{2}}\Delta$ and the constant $C_{R_{2}}$
does not depend on $(\Delta,s,a)$. Then we let $\Delta$ go to zero
along the subsequence $\{\Delta_{k}\}$, since $\lim_{k\rightarrow\infty}V_{\Delta_{k}}^{*}(\cdot)=V^{*}(\cdot)$
and $\lim_{k\rightarrow\infty}R_{2}(\Delta_{k}|s,a)=0$, we obtain
by standard analysis techniques that $V^{*}(\cdot)$ is the solution
of the continuous-time HJB satisfied by $V_{0}^{*}(\cdot)$. Finally,
by the uniqueness of the solution to the equation (\ref{HJB for Value continuous}),
we obtain that $V^{*}(\cdot)=V_{0}^{*}(\cdot)$, and thus $\lim_{\Delta\rightarrow0}V_{\Delta}^{*}(\cdot)=V_{0}^{*}(\cdot)$.

Step 2 --- Prove the identity $\pi_{\Delta}^{\ast}(\cdot)=\pi_{0}^{\ast}(\cdot)$
for sufficiently small $\Delta$

We prove that, there exists $\Delta_{\pi}>0$ such that for any $\Delta<\Delta_{\pi}$,
it holds that $\pi_{\Delta}^{\ast}(s)=\pi_{0}^{\ast}(s)$ for any
$s\in\mathcal{S}$. We know that $\pi_{\Delta}^{\ast}(s)=\arg\max_{a\in\mathcal{A}}Q_{\Delta}^{*}(s,a)$
and the optimal Q function satisfies that
\[
Q_{\Delta}^{*}(s,a)=f(s,a)\Delta+e^{-\gamma\Delta}\sum_{s^{\prime}\in\mathcal{S}}p_{\Delta}(s^{\prime}|s,a)V_{\Delta}^{*}(s^{\prime}).
\]
 Similar to the proof in Step 1, using the formulas of $p_{\Delta}(s^{\prime}|s,a)$,
we have that
\[
\lim_{\Delta\rightarrow0}\frac{1}{\Delta}(Q_{\Delta}^{*}(s,a)-e^{-\gamma\Delta}V_{\Delta}^{*}(s))=f(s,a)+\sum_{s^{\prime}\in\mathcal{S},s^{\prime}\neq s}\lambda_{s^{\prime}}(s,a)V_{0}^{*}(s^{\prime})-\lambda(s,a)V_{0}^{*}(s),
\]
where the right hand side is exactly same with that of the equation
(\ref{HJB for Value continuous}) and the maximizer of the right hand
side is $a=\pi_{0}^{\ast}(s)$. Since $V_{\Delta}^{*}(s)$ is free
of $\Delta$, we have that $\pi_{\Delta}^{\ast}(s)=\arg\max_{a\in\mathcal{A}}Q_{\Delta}^{*}(s,a)=\arg\max_{a\in\mathcal{A}}\frac{1}{\Delta}(Q_{\Delta}^{*}(s,a)-e^{-\gamma\Delta}V_{\Delta}^{*}(s))$,
and thus $\lim_{\Delta\rightarrow0}\pi_{\Delta}^{\ast}(s)=\pi_{0}^{\ast}(s)$.
Since the policies $\pi_{\Delta}^{\ast}(s)$ and $\pi_{0}^{\ast}(s)$
take values in the finite set $\mathcal{A}$, we obtain that $\pi_{\Delta}^{\ast}(s)=\pi_{0}^{\ast}(s)$
for any $s\in\mathcal{S}$ and sufficiently small $\Delta$.

Step 3 --- Prove the value function error order $||V_{\Delta}^{*}(\cdot)-V_{0}^{*}(\cdot)||=O(\Delta)$

For simplicity of notations, we denote by $S_{i}^{*}$ (resp. $S_{t}^{*}$)
the discrete-time (resp. continuous-time) Markov chain under the optimal
policy $a_{i}=\pi_{\Delta}^{\ast}(S_{i})$ (resp. $a_{t}=\pi_{0}^{\ast}(S_{t})$).
In particular, since $\pi_{0}^{\ast}$ is a Markov policy, we have
that $S_{t}^{*}$ is a continuous-time Markov chain, and the rate
parameters are given by $\lambda(s,\pi_{0}^{\ast}(s))$ and $\{\lambda_{s^{\prime}}(s,\pi_{0}^{\ast}(s))\}_{s^{\prime}\in\mathcal{S},s^{\prime}\neq s}$.

Step 3.1 --- 

By the Bellman equations, we have that
\[
V_{\Delta}^{*}(s)=f(s,\pi_{\Delta}^{\ast}(s))\Delta+e^{-\gamma\Delta}\sum_{s^{\prime}\in\mathcal{S}}p_{\Delta}(s^{\prime}|s,\pi_{\Delta}^{\ast}(s))V_{\Delta}^{*}(s^{\prime}),
\]
and by the time-homogeneous Markov property of the continuous-time
policy, we have that
\[
V_{0}^{*}(s)=E[\int_{0}^{\Delta}e^{-\gamma t}f(S_{t}^{*},\pi_{0}^{\ast}(S_{t}^{*}))dt|S_{0}=s]+e^{-\gamma\Delta}\sum_{s^{\prime}\in\mathcal{S}}p(\Delta,s^{\prime}|s,\pi_{0}^{\ast}(s))V_{0}^{*}(s^{\prime}),
\]
where $p(\Delta,s^{\prime}|s,\pi_{0}^{\ast}(s))$ is the conditional
density of $S_{\Delta}^{*}$ given that $S_{t}^{*}=s$. We consider
the case where $\Delta<\Delta_{\pi}$ so that $\pi_{\Delta}^{\ast}(s)=\pi_{0}^{\ast}(s)$.
Then by the properties of continuous-time Markov chain, we have that
\[
p(\Delta,s^{\prime}|s,\pi_{0}^{\ast}(s))=\lambda_{s^{\prime}}(s,\pi_{0}^{\ast}(s))\Delta+O(\Delta^{2})=p_{\Delta}(s^{\prime}|s,\pi_{\Delta}^{\ast}(s))+O(\Delta^{2})\qquad\text{for}\ s^{\prime}\neq s,
\]
and
\[
p(\Delta,s|s,\pi_{0}^{\ast}(s))=1-\lambda(s,\pi_{0}^{\ast}(s))\Delta+O(\Delta^{2})=p_{\Delta}(s|s,\pi_{\Delta}^{\ast}(s))+O(\Delta^{2}).
\]
So we obtain that
\begin{align*}
 & \left|\sum_{s^{\prime}\in\mathcal{S}}p_{\Delta}(s^{\prime}|s,\pi_{\Delta}^{\ast}(s))V_{\Delta}^{*}(s^{\prime})-\sum_{s^{\prime}\in\mathcal{S}}p(s^{\prime}|s,\pi_{0}^{\ast}(s))V_{0}^{*}(s^{\prime})\right|\\
 & =\left|\sum_{s^{\prime}\in\mathcal{S}}p_{\Delta}(s^{\prime}|s,\pi_{\Delta}^{\ast}(s))(V_{\Delta}^{*}(s^{\prime})-V_{0}^{*}(s^{\prime}))+R_{a}(\Delta|s)\right|\\
 & \leq||V_{\Delta}^{*}(\cdot)-V_{0}^{*}(\cdot)||+C_{R_{a}}\Delta^{2},
\end{align*}
where $|R_{a}(\Delta|s)|\leq C_{R_{a}}\Delta^{2}$ and the constant
$C_{R_{a}}$ does not depend on $(\Delta,s)$.

Step 3.2 --- 

For the conditional expectation term, we have that
\begin{align*}
 & E[\int_{0}^{\Delta}e^{-\gamma t}f(S_{t}^{*},\pi_{0}^{\ast}(S_{t}^{*}))dt|S_{0}=s]\\
 & =E[\int_{0}^{\Delta}f(S_{t}^{*},\pi_{0}^{\ast}(S_{t}^{*}))dt|S_{0}=s]+R_{b}(\Delta|s)\\
 & =f(s,\pi_{0}^{\ast}(s))\Delta+E[\int_{0}^{\Delta}f(S_{t}^{*},\pi_{0}^{\ast}(S_{t}^{*}))-f(s,\pi_{0}^{\ast}(s))dt|S_{0}=s]+R_{b}(\Delta|s),
\end{align*}
where $|R_{b}(\Delta|s)|\leq\int_{0}^{\Delta}C_{B}^{f}(1-e^{-\gamma t})dt\leq C_{B}^{f}(\Delta+\frac{1}{\gamma}(e^{-\gamma\Delta}-1))dt\leq2\gamma C_{B}^{f}\Delta^{2}$
holds for sufficiently small $\Delta$. Next, by the Lipschitz property
of $f$, we have that
\begin{align*}
 & |E[\int_{0}^{\Delta}f(S_{t}^{*},\pi_{0}^{\ast}(S_{t}^{*}))-f(s,\pi_{0}^{\ast}(s))dt|S_{0}=s]|\\
 & \leq\int_{0}^{\Delta}C_{Lip}^{f}E[|S_{t}^{*}-s|+|\pi_{0}^{\ast}(S_{t}^{*})-\pi_{0}^{\ast}(s)||S_{0}=s]dt\\
 & =\int_{0}^{\Delta}C_{Lip}^{f}\sum_{s^{\prime}\in\mathcal{S},s^{\prime}\neq s}p(t,s^{\prime}|s,\pi_{0}^{\ast}(s))(|s^{\prime}-s|+|\pi_{0}^{\ast}(s^{\prime})-\pi_{0}^{\ast}(s)|)dt\\
 & \leq\int_{0}^{\Delta}C_{Lip}^{f}\sum_{s^{\prime}\in\mathcal{S},s^{\prime}\neq s}p(t,s^{\prime}|s,\pi_{0}^{\ast}(s))C_{\mathcal{S},\mathcal{A}}dt,
\end{align*}
where $p(t,s^{\prime}|s,\pi_{0}^{\ast}(s))$ is the conditional density
of $S_{t}^{*}$ given that $S_{t}^{*}=s$, and the last inequality
uses the fact that there exists a constant $C_{\mathcal{S},\mathcal{A}}$
such that $|s^{\prime}-s|+|a^{\prime}-a|<C_{\mathcal{S},\mathcal{A}}$
since both $\mathcal{S}$ and $\mathcal{A}$ are finite spaces. Then,
by the properties of continuous-time Markov chain we have that
\[
\sum_{s^{\prime}\in\mathcal{S},s^{\prime}\neq s}p(t,s^{\prime}|s,\pi_{0}^{\ast}(s))=\sum_{s^{\prime}\in\mathcal{S},s^{\prime}\neq s}\lambda_{s^{\prime}}(s,\pi_{0}^{\ast}(s))t+O(t^{2})=\lambda(s,\pi_{0}^{\ast}(s))t+O(t^{2}).
\]
So we obtain that, there exists a constant $C_{2}^{\lambda}>0$ such
that
\begin{align*}
 & |E[\int_{0}^{\Delta}f(S_{t}^{*},\pi_{0}^{\ast}(S_{t}^{*}))-f(s,\pi_{0}^{\ast}(s))dt|S_{0}=s]|\\
 & \leq C_{Lip}^{f}C_{\mathcal{S},\mathcal{A}}\int_{0}^{\Delta}(\lambda(s,\pi_{0}^{\ast}(s))t+C_{2}^{\lambda}t^{2})dt\\
 & \leq C_{Lip}^{f}C_{\mathcal{S},\mathcal{A}}(\frac{1}{2}\lambda(s,\pi_{0}^{\ast}(s))\Delta^{2}+\frac{1}{3}C_{2}^{\lambda}\Delta^{3}).
\end{align*}
Thus, there exists a constant $C_{2}$ such that $|E[\int_{0}^{\Delta}f(S_{t}^{*},\pi_{0}^{\ast}(S_{t}^{*}))-f(s,\pi_{0}^{\ast}(s))dt|S_{0}=s]|\leq C_{2}\Delta^{2}$
for sufficiently small $\Delta$.

Step 3.3 --- 

Substracting the Bellman equations satisfied by $V_{\Delta}^{*}(s)$
and $V_{0}^{*}(s)$, using the inequalities we obtained in Step 3.2
and 3.3, we have that, for any sufficiently small $\Delta$,
\begin{align*}
|V_{\Delta}^{*}(s)-V_{0}^{*}(s)| & \leq|E[\int_{0}^{\Delta}e^{-\gamma t}f(S_{t}^{*},\pi_{0}^{\ast}(S_{t}^{*}))dt|S_{0}=s]-f(s,\pi_{\Delta}^{\ast}(s))\Delta|\\
 & +e^{-\gamma\Delta}\left|\sum_{s^{\prime}\in\mathcal{S}}p_{\Delta}(s^{\prime}|s,\pi_{\Delta}^{\ast}(s))V_{\Delta}^{*}(s^{\prime})-\sum_{s^{\prime}\in\mathcal{S}}p(s^{\prime}|s,\pi_{0}^{\ast}(s))V_{0}^{*}(s^{\prime})\right|\\
 & \leq C_{2}\Delta^{2}+e^{-\gamma\Delta}(||V_{\Delta}^{*}(\cdot)-V_{0}^{*}(\cdot)||+C_{R_{a}}\Delta^{2}),
\end{align*}
where the norm is given by $||V_{\Delta}^{*}(\cdot)-V_{0}^{*}(\cdot)||=\max_{s\in\mathcal{S}}|V_{\Delta}^{*}(s)-V_{0}^{*}(s)|$.
So we obtain that
\[
||V_{\Delta}^{*}(\cdot)-V_{0}^{*}(\cdot)||\leq C_{2}\Delta^{2}+e^{-\gamma\Delta}(||V_{\Delta}^{*}(\cdot)-V_{0}^{*}(\cdot)||+C_{R_{a}}\Delta^{2}),
\]
and thus 
\[
||V_{\Delta}^{*}(\cdot)-V_{0}^{*}(\cdot)||\leq\frac{1}{1-e^{-\gamma\Delta}}(C_{2}+e^{-\gamma\Delta}C_{R_{a}})\Delta^{2}.
\]
Since $1-e^{-\gamma\Delta}=\gamma\Delta+O(\Delta^{2})$, we obtain
that, there exists $\Delta_{V}>0$ and a constant $C_{V}>0$ such
that $||V_{\Delta}^{*}(\cdot)-V_{0}^{*}(\cdot)||\leq C_{V}\Delta$
for any $\Delta<\Delta_{V}$.

\subsection{Proof of Theorem \ref{Thm:value discretize error}}
\begin{proof}
To prove Theorem \ref{Thm:value discretize error}, it suffices to
show that our models can be reduced to the cases studied in Section
\ref{sec:general MDP convergence results}.

First, we show that the state process $S_{t}=(X_{t},Y_{t})$ is a
controlled continuous-time Markov chain on finite space. It suffices
to show this for $Y_{t}=-N_{t}^{a}+N_{t}^{b}$. Though $N_{t}^{a}$
and $N_{t}^{b}$ are Poisson processes, the rate parameter of $N_{t}^{a}$
(resp. $N_{t}^{b}$) is set to zero when $Y_{t}=-N_{Y}$ (resp. $Y_{t}=N_{Y}$)
according to our rules for the action variable. So, $Y_{t}$ is a
controlled continuous-time Markov chain on $\mathcal{S}_{Y}=\{-N_{Y},...,-2,-1,0,1,2,...,N_{Y}\}$.
By our construction, the rate parameters for $S_{t}$ are uniformly
bounded on the finite space $\mathcal{S}\times\mathcal{A}$.

Then, we show that the objective functions in our continuous-time
MDP $\mathcal{M}_{0}$ and discrete-time MDP $\mathcal{M}_{\Delta}$
can be expressed as the forms in Section \ref{sec:general MDP convergence results}.
For simplicity of notations, we omit the conditions $S_{0}=s$ in
the conditional expectations. All the following computations involving
the interchange between expectation and infinite series or integrals
are rigorous owing to the uniformly boundness of the functions therein.
By the properties of the continuous-time Markov chain and the expectation
rules for the stochastic integral w.r.t. the Poisson process (e.g.,
Section 3.2 in \citet{hanson2007BookJumpDiffusion}), we have that
\begin{align*}
V_{0}^{\pi}(s) & =E[\int_{0}^{+\infty}e^{-\gamma t}\left\{ 1_{\{Y_{t}>-N_{Y}\}}\cdot(p_{t}^{a}-X_{t}-c)dN_{t}^{a}+1_{\{Y_{t}<N_{Y}\}}\cdot(X_{t}-p_{t}^{b}-c)dN_{t}^{b}+Y_{t}dX_{t}-\psi(Y_{t})dt\right\} ]\\
 & =E[\int_{0}^{+\infty}e^{-\gamma t}\{1_{\{Y_{t}>-N_{Y}\}}(p_{t}^{a}-X_{t}-c)\cdot\Lambda(p_{t}^{a}-X_{t})dt+1_{\{Y_{t}<N_{Y}\}}(X_{t}-p_{t}^{b}-c)\cdot\Lambda(X_{t}-p_{t}^{b})dt\\
 & -\psi(Y_{t})dt+Y_{t}\mu(X_{t},a_{t})dt\}]\\
 & =E[\int_{0}^{+\infty}e^{-\gamma t}f(S_{t},a_{t})dt],
\end{align*}
where, for any $x=\frac{1}{2}k\delta_{P}\in\mathcal{S}_{X}$ and $k\in\{1,2,\ldots,|\mathcal{S}_{X}|\}$,
the function $\mu$ is defined as $\mu(x,a):=(\lambda_{k,k+1}(a)-\lambda_{k,k-1}(a))\frac{1}{2}\delta_{P}$
with $\lambda_{1,0}(a)=\lambda_{|\mathcal{S}_{X}|,|\mathcal{S}_{X}|+1}(a)=0$,
and the function $f$ is defined as
\[
f(s,a):=1_{\{y>-N_{Y}\}}\cdot(p^{a}-x-c)\Lambda(p^{a}-x)+1_{\{y<N_{Y}\}}\cdot(x-p^{b}-c)\Lambda(x-p^{b})-\psi(y)+y\mu(x,a),
\]
for the variables $s=(x,y),a=(p^{a},p^{b})$. Then by the transition
probabilities of the discrete-time Markov chains in the MDP $\mathcal{M}_{\Delta}$,
we have that
\begin{align*}
V_{\Delta}^{\pi}(s) & =E[\sum_{i=0}^{+\infty}e^{-i\gamma\Delta}[(p_{i}^{a}-X_{i}-c)n_{i}^{a}\cdot1_{\{Y_{i}>-N_{Y}\}}+(X_{i}-p_{i}^{b}-c)n_{i}^{b}\cdot1_{\{Y_{i}<N_{Y}\}}+(X_{i+1}-X_{i})Y_{i}-\psi(Y_{i})\Delta]]\\
 & =E[\sum_{i=0}^{+\infty}e^{-i\gamma\Delta}[1_{\{Y_{i}>-N_{Y}\}}(p_{i}^{a}-X_{i}-c)\Lambda(p_{i}^{a}-X_{i})+1_{\{Y_{i}<N_{Y}\}}(X_{i}-p_{i}^{b}-c)\Lambda(X_{i}-p_{i}^{b})\\
 & -\psi(Y_{i})+Y_{i}\mu(X_{i},a_{i})]\Delta]\\
 & =E[\sum_{i=0}^{+\infty}e^{-i\gamma\Delta}f(S_{i},a_{i})\Delta],
\end{align*}
where the functions $f$ and $\lambda_{X}$ are the same with those
in continuous-time MDP. Clearly, the function $f$ is uniformly bounded
on the finite space $\mathcal{S}\times\mathcal{A}$.

Thus, the results in Theorem \ref{Thm:value discretize error} follow
from the general results in Lemma \ref{Lemma:value discretize error general}.
\end{proof}

\section{Proof of Theorem \ref{Thm:sample cplx bound Delta}}

Consider the sample complexity defined using Q value, which is the
number of iteration steps $n$ such that, with probability at least
$1-\delta$, it holds $||Q_{\Delta}^{(n)}-Q_{\Delta}^{*}||\leq\varepsilon_{Q}$,
where the norm is defined as $||Q_{\Delta}^{(n)}-Q_{\Delta}^{*}||=\max_{(s,a)\in\mathcal{S}\times\mathcal{A}}|Q_{\Delta}^{(n)}(s,a)-Q_{\Delta}^{*}(s,a)|$.
Using Theorem 4 in \citet{JMLR2003Qcomplexity}, we obtain that the
upper bound for the sample complexity is given by $n=\Omega(B_{n})$
and
\begin{equation}
B_{n}=\left(L^{1+3\omega}(C_{B}^{V})^{2}\frac{1}{\frac{1}{4}(1-e^{-\gamma\Delta})^{2}\varepsilon_{Q}^{2}}\log(\frac{|\mathcal{S}||\mathcal{A}|C_{B}^{V}}{\delta\frac{1}{2}(1-e^{-\gamma\Delta})\varepsilon_{Q}})\right)^{\frac{1}{\omega}}+\left(L\frac{1}{\frac{1}{2}(1-e^{-\gamma\Delta})}\log(\frac{C_{B}^{V}}{\varepsilon_{Q}})\right)^{\frac{1}{1-\omega}},\label{sample cplx original JMLR}
\end{equation}
where $C_{B}^{V}$ is the upper bound for the value function $V_{\Delta}^{\pi}(s)$,
and $L$ is the covering time such that, from any start state, with
probability at least $\frac{1}{2}$, all state-action pairs appear
in the sequence within $L$ steps. In what follows, we first show
that in our high-frequency setup, the Q function error bound $||Q_{\Delta}^{(n)}-Q_{\Delta}^{*}||\leq\varepsilon_{Q}$
implies the value function error bound $||V_{\Delta}^{(n)}-V_{\Delta}^{*}||\leq\varepsilon_{Q}$.
Then, we express the bound (\ref{sample cplx original JMLR}) using
our model parameters.

Indeed, by Theorem \ref{Thm:value discretize error}, for sufficiently
small $\Delta$, we have that $\pi_{\Delta}^{*}(s)=\pi^{*}(s)$ for
any $s\in\mathcal{S}$. Then, as $\varepsilon_{Q}$ is small enough
such that $\max_{(s,a)\in\mathcal{S}\times\mathcal{A}}|Q_{\Delta}^{(n)}(s,a)-Q_{\Delta}^{*}(s,a)|\leq\varepsilon_{Q}$
implies that $\text{arg}\max_{a\in\mathcal{A}}Q_{\Delta}^{(n)}(s,a)=\text{arg}\max_{a\in\mathcal{A}}Q_{\Delta}^{*}(s,a)$
for any $s\in\mathcal{S}$, we have that the policy $\pi_{\Delta}^{(n)}(s)=\text{arg}\max_{a\in\mathcal{A}}Q_{\Delta}^{(n)}(s,a)$
learned at the $n$th iteration is the same as the true optimal policy
$\pi_{\Delta}^{*}(s)=\text{arg}\max_{a\in\mathcal{A}}Q_{\Delta}^{*}(s,a)$.
So, we have that $|V_{\Delta}^{(n)}(s)-V_{\Delta}^{*}(s)|=|Q_{\Delta}^{(n)}(s,\pi_{\Delta}^{(n)}(s))-Q_{\Delta}^{*}(s,\pi_{\Delta}^{*}(s))|=|Q_{\Delta}^{(n)}(s,\pi_{\Delta}^{*}(s))-Q_{\Delta}^{*}(s,\pi_{\Delta}^{*}(s))|\leq\varepsilon_{Q}$
for any $s\in\mathcal{S}$, and thus $||V_{\Delta}^{(n)}-V_{\Delta}^{*}||\leq\varepsilon_{Q}$.

The upper bound $C_{B}^{V}$ is given by $C_{B}^{V}=\frac{1}{1-e^{-\gamma\Delta}}C_{B}^{f}$,
where $C_{B}^{f}$ is the upper bound for the cost function. We now
derive an upper bound $L$ for the $\varepsilon$-greedy exploration.
By the structure of the transition probability matrix of $X_{i}$,
we have that, there exists a constant $c_{S}>0$ such that, for any
policy, it holds 
\[
P(S_{i+(|\mathcal{S}_{X}|+|\mathcal{S}_{Y}|-2)}=s^{\prime}|S_{i}=s)\geq c_{S}\Delta^{|\mathcal{S}_{X}|+|\mathcal{S}_{Y}|-2}.
\]
Under the $\varepsilon$-greedy policy, we have that, for any $(s^{\prime},s,a)\in\mathcal{S}\times\mathcal{S}\times\mathcal{A}$,
it holds
\[
P(S_{i+(|\mathcal{S}_{X}|+|\mathcal{S}_{Y}|-2)}=s^{\prime},a_{i+(|\mathcal{S}_{X}|+|\mathcal{S}_{Y}|-2)}=a|S_{i}=s)\geq\frac{\varepsilon_{0}}{|\mathcal{A}|}c_{S}\Delta^{|\mathcal{S}_{X}|+|\mathcal{S}_{Y}|-2},
\]
where $\varepsilon_{0}$ is the smallest value of the $\varepsilon$
sequence used in the $\varepsilon$-greedy exploration. Denote by
$\tilde{L}$ as the smallest number $n$ such that the random variable
such that all state-action pairs appear in the sequence of length
$n$. Then similar to the proof of Proposition 1 in \citet{Qnn2018coverTimeLemma},
we have that there exists a constant $C_{L}>0$ such that $E[\tilde{L}]\leq C_{L}\frac{1}{\varepsilon_{0}c_{S}}\Delta^{2-(|\mathcal{S}_{X}|+|\mathcal{S}_{Y}|)}(|\mathcal{S}_{X}|+|\mathcal{S}_{Y}|)|\mathcal{A}|\log(|\mathcal{S}||\mathcal{A}|)$
for any start state. By the Markov inequality, we have that $P(\tilde{L}<L)>\frac{1}{2}$
where
\[
L=2C_{L}\frac{1}{\varepsilon_{0}c_{S}}\Delta^{2-(|\mathcal{S}_{X}|+|\mathcal{S}_{Y}|)}(|\mathcal{S}_{X}|+|\mathcal{S}_{Y}|)|\mathcal{A}|\log(|\mathcal{S}||\mathcal{A}|).
\]

Plugging the formulae of $C_{B}^{V}$ and $L$ into (\ref{sample cplx original JMLR})
and suppressing all logarithmic factors, we obtain that
\begin{align*}
B_{n} & =((|\mathcal{S}_{X}|+|\mathcal{S}_{Y}|)|\mathcal{A}|)^{3+\frac{1}{\omega}}\left(\varepsilon_{0}^{-(1+3\omega)}\Delta^{(2-|\mathcal{S}_{X}|-|\mathcal{S}_{Y}|)(1+3\omega)}\gamma^{-2}\Delta^{-2}\frac{1}{\gamma^{2}\Delta^{2}\varepsilon_{Q}^{2}}\right)^{\frac{1}{\omega}}\\
 & +\left(\frac{1}{\varepsilon_{0}}\Delta^{2-|\mathcal{S}_{X}|-|\mathcal{S}_{Y}|}\frac{1}{\gamma}(|\mathcal{S}_{X}|+|\mathcal{S}_{Y}|)|\mathcal{A}|\frac{1}{\Delta}\right)^{\frac{1}{1-\omega}}\\
 & =((|\mathcal{S}_{X}|+|\mathcal{S}_{Y}|)|\mathcal{A}|\varepsilon_{0}^{-1})^{3+\frac{1}{\omega}}\gamma^{-\frac{4}{\omega}}\Delta^{-\frac{4}{\omega}+2(3+\frac{1}{\omega})-|\mathcal{S}_{X}|(3+\frac{1}{\omega})-|\mathcal{S}_{Y}|(3+\frac{1}{\omega})}\varepsilon_{Q}^{-\frac{2}{\omega}}\\
 & +((|\mathcal{S}_{X}|+|\mathcal{S}_{Y}|)|\mathcal{A}|\varepsilon_{0}^{-1})^{\frac{1}{1-\omega}}\gamma^{-\frac{1}{1-\omega}}\Delta^{(1-|\mathcal{S}_{X}|-|\mathcal{S}_{Y}|)\frac{1}{1-\omega}}.
\end{align*}
So we finally obtain that 
\begin{align*}
B_{n} & =((|\mathcal{S}_{X}|+|\mathcal{S}_{Y}|)|\mathcal{A}|\varepsilon_{0}^{-1})^{3+\frac{1}{\omega}}\varepsilon_{Q}^{-\frac{2}{\omega}}\gamma^{-\frac{4}{\omega}}\Delta^{6-\frac{2}{\omega}-(|\mathcal{S}_{X}|+|\mathcal{S}_{Y}|)(3+\frac{1}{\omega})}\\
 & +((|\mathcal{S}_{X}|+|\mathcal{S}_{Y}|)|\mathcal{A}|\varepsilon_{0}^{-1})^{\frac{1}{1-\omega}}\gamma^{-\frac{1}{1-\omega}}\Delta^{(1-|\mathcal{S}_{X}|-|\mathcal{S}_{Y}|)\frac{1}{1-\omega}}.
\end{align*}

\section{Proof of Theorem \ref{Thm:2player Nash converge}}

\subsection{General results for controlled Markov chains\label{sec:general MDP game convergence}}

To prove our convergence results, we consider a more general case
for continuous-time Markov chains as follows, which is potentially
useful for future research.

For the continuous-time game $\mathcal{G}_{0}$, the state variable
$S_{t}$ is a controlled Markov chain with finite state space $\mathcal{S}$.
For any $s\in\mathcal{S}$, the transition rate parameters are given
by $\lambda(s^{\prime}|s,a^{1},a^{2})$ satisfying that $\sum_{s^{\prime}\in\mathcal{S},s^{\prime}\neq s}\lambda(s^{\prime}|s,a^{1},a^{2})=\lambda(s|s,a^{1},a^{2})$
and $\lambda(s^{\prime}|s,a^{1},a^{2})\geq0$, where $a_{k}$ is the
action variable of the $k$-th player. More precisely, when $S_{t}=s$
and $(a_{t}^{1},a_{t}^{2})=(a^{1},a^{2})$ at time $t$, then after
a random waiting time $\tau$ which follows an exponential distribution
with parameter $\lambda(s|s,a^{1},a^{2})$, there will be a state
transition in $S_{t}$ and the probability is given by $P(S_{t+\tau}=s^{\prime})=\lambda(s^{\prime}|s,a^{1},a^{2})/\lambda(s|s,a^{1},a^{2})$.
We denote by $S_{t}^{\pi^{1},\pi^{2}}$ the state variable under the
pair of strategies $(\pi^{1},\pi^{2})$. The value function of the
$k$-th player is given by
\[
V_{0}^{k,\pi^{1},\pi^{2}}(s):=E\left[\left.\int_{0}^{+\infty}e^{-\gamma t}r^{k}(S_{t}^{\pi^{1},\pi^{2}},a_{t}^{1},a_{t}^{2})dt\right\vert S_{0}=s\right],
\]
where the action $a_{t}^{k}$ is chosen under $\pi^{k}$ for $k=1,2$.
The strategy sets (i.e., the randomized Markov strategies $\Pi_{M}^{k}$,
the stationary strategies $\Pi_{s}^{k}$, and the admissible strategies
$\Pi^{k}$) for the two players and the Nash equilibrium are defined
in the same manner as we did in Section \ref{sec:continuous MM game}.

For the discrete-time game $\mathcal{G}_{\Delta}$, the state variable
$S_{i}$ is a discrete-time controlled Markov chain. When the two
players take actions $(a^{1},a^{2})$ at time $i$, the transition
probabilities are given by $P(S_{i+1}=s)=p_{\Delta}(s|s,a^{1},a^{2})=1-\lambda(s|s,a^{1},a^{2})\Delta+O(\Delta^{2})$,
and $P(S_{i+1}=s^{\prime})=p_{\Delta}(s^{\prime}|s,a^{1},a^{2})=\lambda(s^{\prime}|s,a^{1},a^{2})\Delta+O(\Delta^{2})$
for any $s^{\prime}\in\mathcal{S}$ and $s^{\prime}\neq s$, where
we assume that $\Delta$ is sufficiently small so that all these probabilities
are in $[0,1]$. We denote by $S_{i}^{\pi^{1},\pi^{2}}$ the state
variable under the pair of strategies $(\pi^{1},\pi^{2})$. The value
function of the $k$-th player is given by
\[
V_{\Delta}^{k,\pi^{1},\pi^{2}}(s):=E\left[\left.\sum_{i=0}^{+\infty}e^{-i\gamma\Delta}r^{k}(S_{i}^{\pi^{1},\pi^{2}},a_{i}^{1},a_{i}^{2})\Delta\right\vert S_{0}=s\right],
\]
where the action $a_{t}^{k}$ is chosen under $\pi^{k}$ for $k=1,2$.
The strategy sets (i.e., the randomized Markov strategies $\Pi_{M}^{k}$,
the stationary strategies $\Pi_{s}^{k}$, and the admissible strategies
$\Pi^{k}$) for the two players and the Nash equilibrium are defined
in the same manner as we did in Section \ref{sec:continuous MM game}.

We assume that the function $\lambda(s^{\prime}|s,a^{1},a^{2})$ (resp.
$r^{k}(s,a^{1},a^{2})$) are uniformly bounded over $\mathcal{S}\times\mathcal{S}\times\mathcal{A}\times\mathcal{A}$
(resp. $\mathcal{S}\times\mathcal{A}\times\mathcal{A}$). Then we
have that
\begin{lem}
\label{Lemma:Game converge general}There exist pairs of stationary
strategies $(\pi_{0}^{1,*},\pi_{0}^{2,*}),(\pi_{\Delta}^{1,*},\pi_{\Delta}^{2,*})\in\Pi_{s}^{1}\times\Pi_{s}^{2}$
such that $(\pi_{0}^{1,*},\pi_{0}^{2,*})$ (resp. $(\pi_{\Delta}^{1,*},\pi_{\Delta}^{2,*})$)
is a Nash equilibrium in the continuous-time (resp. discrete-time)
game $\mathcal{G}_{0}$ (resp. $\mathcal{G}_{\Delta}$). Moreover,
assuming the uniqueness of the Nash equilibrium $(\pi_{0}^{1,*},\pi_{0}^{2,*})$
(resp. $(\pi_{\Delta}^{1,*},\pi_{\Delta}^{2,*})$) in the game $\mathcal{G}_{0}$
(resp. $\mathcal{G}_{\Delta}$), then we have the following convergence
results: (i) The policies satisfy that $||\pi_{\Delta}^{k,*}(\cdot|s)-\pi_{0}^{k,*}(\cdot|s)||\rightarrow0$
as $\Delta\rightarrow0$ for $k=1,2$, where the norm is defined as
$||\pi(\cdot)-\pi^{\prime}(\cdot)||:=\max_{a\in\mathcal{A}}|\pi(a)-\pi^{\prime}(a)|$.
(ii) The value functions satisfy that $|V_{\Delta}^{k,\pi_{\Delta}^{1,*},\pi_{\Delta}^{2,*}}(s)-V_{0}^{k,\pi_{0}^{1,*},\pi_{0}^{2,*}}(s)|\rightarrow0$
as $\Delta\rightarrow0$ for any $s\in\mathcal{S}$ and $k=1,2$.
\end{lem}

\begin{proof}
The existence of the equilibrium point $(\pi_{0}^{1,*},\pi_{0}^{2,*})$
(resp. $(\pi_{\Delta}^{1,*},\pi_{\Delta}^{2,*})$) in stationary strategies
$\Pi_{s}^{1}\times\Pi_{s}^{2}$ for the continuous-time (resp. discrete-time)
game is guaranteed by Theorem 5.1 in \citet{GuoHernz2005continuousMDPnonzero_sum_game}
(resp. Theorem 4.6.4 in \citet{FilarVrieze2012Book_MDP_game}).

Under the assumption on the uniqueness of equilibrium point, we obtain
by Theorem 5.1 in \citet{GuoHernz2005continuousMDPnonzero_sum_game}
that, for the continuous-time game $\mathcal{G}_{0}$, $(\pi_{0}^{1,*},\pi_{0}^{2,*})$
is the unique pair of strategies that satisfies the following pair
of optimality equations:
\begin{align}
\gamma V_{0}^{1,\pi_{0}^{1,*},\pi_{0}^{2,*}}(s) & =\sup_{\pi^{1}\in\Pi^{1}}\bigl(r^{1}(s,\pi^{1},\pi_{0}^{2,*})+\sum_{s^{\prime}\in\mathcal{S}}\lambda(s^{\prime}|s,\pi^{1},\pi_{0}^{2,*})V_{0}^{1,\pi_{0}^{1,*},\pi_{0}^{2,*}}(s^{\prime})\bigr)\nonumber \\
 & =\max_{a\in\mathcal{A}}\bigl(r^{1}(s,a,\pi_{0}^{2,*})+\sum_{s^{\prime}\in\mathcal{S}}\lambda(s^{\prime}|s,a,\pi_{0}^{2,*})V_{0}^{1,\pi_{0}^{1,*},\pi_{0}^{2,*}}(s^{\prime})\bigr),\label{cont V1 game Bellman equation}
\end{align}
and
\begin{align}
\gamma V_{0}^{2,\pi_{0}^{1,*},\pi_{0}^{2,*}}(s) & =\sup_{\pi^{2}\in\Pi^{2}}\bigl(r^{2}(s,\pi_{0}^{1,*},\pi^{2})+\sum_{s^{\prime}\in\mathcal{S}}\lambda(s^{\prime}|s,\pi_{0}^{1,*},\pi^{2})V_{0}^{2,\pi_{0}^{1,*},\pi_{0}^{2,*}}(s^{\prime})\bigr)\nonumber \\
 & =\max_{a\in\mathcal{A}}\bigl(r^{2}(s,\pi_{0}^{1,*},a)+\sum_{s^{\prime}\in\mathcal{S}}\lambda(s^{\prime}|s,\pi_{0}^{1,*},a)V_{0}^{2,\pi_{0}^{1,*},\pi_{0}^{2,*}}(s^{\prime})\bigr).\label{cont V2 game Bellman equation}
\end{align}
Here, the equalities in (\ref{cont V1 game Bellman equation}) and
(\ref{cont V2 game Bellman equation}) follow from the fact that the
action space $\mathcal{A}$ is finite.

We obtain by (ii) in Theorem 4.6.5 in \citet{FilarVrieze2012Book_MDP_game}
that, for the discrete-time game $\mathcal{G}_{\Delta}$, the Nash
equilibrium point $(\pi_{\Delta}^{1,*},\pi_{\Delta}^{2,*})$ satisfies
the following pair of optimality equations:
\[
V_{\Delta}^{1,\pi_{\Delta}^{1,*},\pi_{\Delta}^{2,*}}(s)=\sup_{\pi^{1}\in\Pi_{\Delta}^{1}}\bigl(r^{1}(s,\pi^{1},\pi_{\Delta}^{2,*})\Delta+e^{-\gamma\Delta}\sum_{s^{\prime}\in\mathcal{S}}p_{\Delta}(s^{\prime}|s,\pi^{1},\pi_{\Delta}^{2,*})V_{\Delta}^{1,\pi_{\Delta}^{1,*},\pi_{\Delta}^{2,*}}(s^{\prime})\bigr),
\]
and
\[
V_{\Delta}^{2,\pi_{\Delta}^{1,*},\pi_{\Delta}^{2,*}}(s)=\sup_{\pi^{2}\in\Pi_{\Delta}^{2}}\bigl(r^{2}(s,\pi^{1},\pi_{\Delta}^{2,*})\Delta+e^{-\gamma\Delta}\sum_{s^{\prime}\in\mathcal{S}}p_{\Delta}(s^{\prime}|s,\pi_{\Delta}^{1,*},\pi^{2})V_{\Delta}^{2,\pi_{\Delta}^{1,*},\pi_{\Delta}^{2,*}}(s^{\prime})\bigr).
\]
Next, by definition of the transition probability function $p_{\Delta}$
and the uniformly boundness of the value function $V^{k}$, we have
that
\[
\frac{1}{\Delta}(1-e^{-\gamma\Delta})V_{\Delta}^{1,\pi_{\Delta}^{1,*},\pi_{\Delta}^{2,*}}(s)=\sup_{\pi^{1}\in\Pi_{\Delta}^{1}}\bigl(r^{1}(s,\pi^{1},\pi_{\Delta}^{2,*})+e^{-\gamma\Delta}\sum_{s^{\prime}\in\mathcal{S}}\lambda(s^{\prime}|s,\pi^{1},\pi_{\Delta}^{2,*})V_{\Delta}^{1,\pi_{\Delta}^{1,*},\pi_{\Delta}^{2,*}}(s^{\prime})+R_{1}^{1}(\Delta|s,\pi^{1})\bigr),
\]
and
\[
\frac{1}{\Delta}(1-e^{-\gamma\Delta})V_{\Delta}^{2,\pi_{\Delta}^{1,*},\pi_{\Delta}^{2,*}}(s)=\sup_{\pi^{2}\in\Pi_{\Delta}^{2}}\bigl(r^{2}(s,\pi_{\Delta}^{1,*},\pi^{2})+e^{-\gamma\Delta}\sum_{s^{\prime}\in\mathcal{S}}\lambda(s^{\prime}|s,\pi_{\Delta}^{1,*},\pi^{2})V_{\Delta}^{2,\pi_{\Delta}^{1,*},\pi_{\Delta}^{2,*}}(s^{\prime})+R_{1}^{2}(\Delta|s,\pi^{2})\bigr),
\]
where $|R_{1}^{k}(\Delta|s,\pi^{k})|<C_{R_{1}}\Delta$ for some constant
$C_{R_{1}}$. Then, using the Taylor expansion of $e^{-\gamma\Delta}$,
we obtain that
\begin{align*}
\gamma V_{\Delta}^{1,\pi_{\Delta}^{1,*},\pi_{\Delta}^{2,*}}(s) & =\sup_{\pi^{1}\in\Pi_{\Delta}^{1}}\bigl(r^{1}(s,\pi^{1},\pi_{\Delta}^{2,*})+\sum_{s^{\prime}\in\mathcal{S}}\lambda(s^{\prime}|s,\pi^{1},\pi_{\Delta}^{2,*})V_{\Delta}^{1,\pi_{\Delta}^{1,*},\pi_{\Delta}^{2,*}}(s^{\prime})+R_{2}^{1}(\Delta|s,\pi^{1})\bigr)\\
 & =\max_{a\in\mathcal{A}}\bigl(r^{1}(s,a,\pi_{\Delta}^{2,*})+\sum_{s^{\prime}\in\mathcal{S}}\lambda(s^{\prime}|s,a,\pi_{\Delta}^{2,*})V_{\Delta}^{1,\pi_{\Delta}^{1,*},\pi_{\Delta}^{2,*}}(s^{\prime})+R_{2}^{1}(\Delta|s,a)\bigr),
\end{align*}
and
\begin{align*}
\gamma V_{\Delta}^{2,\pi_{\Delta}^{1,*},\pi_{\Delta}^{2,*}}(s) & =\sup_{\pi^{2}\in\Pi_{\Delta}^{2}}\bigl(r^{2}(s,\pi_{\Delta}^{1,*},\pi^{2})+\sum_{s^{\prime}\in\mathcal{S}}\lambda(s^{\prime}|s,\pi_{\Delta}^{1,*},\pi^{2})V_{\Delta}^{2,\pi_{\Delta}^{1,*},\pi_{\Delta}^{2,*}}(s^{\prime})+R_{2}^{2}(\Delta|s,\pi^{2})\bigr)\\
 & =\max_{a\in\mathcal{A}}\bigl(r^{2}(s,\pi_{\Delta}^{1,*},a)+\sum_{s^{\prime}\in\mathcal{S}}\lambda(s^{\prime}|s,\pi_{\Delta}^{1,*},a)V_{\Delta}^{2,\pi_{\Delta}^{1,*},\pi_{\Delta}^{2,*}}(s^{\prime})+R_{2}^{2}(\Delta|s,a)\bigr),
\end{align*}
where $|R_{2}^{k}(\Delta|s,\pi^{k})|<C_{R_{2}}\Delta$ for some constant
$C_{R_{2}}$, and the last equality is because the action space $\mathcal{A}$
is finite.

Note that $\{(\pi_{\Delta}^{1,*},\pi_{\Delta}^{2,*},V_{\Delta}^{1,\pi_{\Delta}^{1,*},\pi_{\Delta}^{2,*}},,V_{\Delta}^{2,\pi_{\Delta}^{1,*},\pi_{\Delta}^{2,*}})\}_{\Delta>0}$
are sequence indexed by $\Delta$ and are in a compact space. Thus,
for any sequence of $\Delta$ that converges to zero, there always
exists a subsequence, which we denote as $\{\Delta_{k}\}$, such that
$\lim_{k\rightarrow\infty}(\pi_{\Delta}^{1,*},\pi_{\Delta}^{2,*},V_{\Delta}^{1,\pi_{\Delta}^{1,*},\pi_{\Delta}^{2,*}},,V_{\Delta}^{2,\pi_{\Delta}^{1,*},\pi_{\Delta}^{2,*}})=(\pi^{1,*},\pi^{2,*},V^{1,*},V^{2,*})$.
To prove our convergence results, it suffices to prove that
\[
(\pi^{1,*},\pi^{2,*},V^{1,*},V^{2,*})=(\pi_{0}^{1,*},\pi_{0}^{2,*},V_{0}^{1,\pi_{0}^{1,*},\pi_{0}^{2,*}},V_{0}^{2,\pi_{0}^{1,*},\pi_{0}^{2,*}})
\]
 regardless of the choice of the subsequence $\{\Delta_{k}\}$.

Indeed, letting $\Delta$ go to zero along the subsequence $\{\Delta_{k}\}$,
we obtain by the previous pair of equations that
\begin{align*}
\gamma V^{1,*}(s) & =\max_{a\in\mathcal{A}}\bigl(r^{1}(s,a,\pi^{2,*})+\sum_{s^{\prime}\in\mathcal{S}}\lambda(s^{\prime}|s,a,\pi^{2,*})V^{1,*}(s^{\prime})\bigr)\\
 & =\sup_{\pi^{1}\in\Pi^{1}}\bigl(r^{1}(s,\pi^{1},\pi^{2,*})+\sum_{s^{\prime}\in\mathcal{S}}\lambda(s^{\prime}|s,\pi^{1},\pi^{2,*})V^{1,*}(s^{\prime})\bigr),
\end{align*}
and
\begin{align*}
\gamma V^{2,*}(s) & =\max_{a\in\mathcal{A}}\bigl(r^{2}(s,\pi^{1,*},a)+\sum_{s^{\prime}\in\mathcal{S}}\lambda(s^{\prime}|s,\pi^{1,*},a)V^{2,*}(s^{\prime})\bigr)\\
 & =\sup_{\pi^{2}\in\Pi^{2}}\bigl(r^{2}(s,\pi^{1,*},\pi^{2})+\sum_{s^{\prime}\in\mathcal{S}}\lambda(s^{\prime}|s,\pi^{1,*},\pi^{2})V^{2,*}(s^{\prime})\bigr).
\end{align*}
Using the same manner, we obtain that $\gamma V^{k,*}(s)=r^{k}(s,\pi^{1,*},\pi^{2,*})+\sum_{s^{\prime}\in\mathcal{S}}\lambda(s^{\prime}|s,\pi^{1,*},\pi^{2,*})V^{k,*}(s^{\prime})$
for $k=1,2$. Then, we obtain by (a) in Lemma 7.2 of \citet{GuoHernz2005continuousMDPnonzero_sum_game}
that $V_{0}^{k,\pi^{1,*},\pi^{2,*}}$ is the unique solution to the
equation $\gamma V(s)=r^{k}(s,\pi^{1,*},\pi^{2,*})+\sum_{s^{\prime}\in\mathcal{S}}\lambda(s^{\prime}|s,\pi^{1,*},\pi^{2,*})V(s^{\prime})$
for $k=1,2$. Thus, by the above pair of equations for $V^{k,*}$,
we obtain that
\[
V^{1,*}(s)=V_{0}^{1,\pi^{1,*},\pi^{2,*}}(s)\qquad\text{and}\qquad V^{2,*}(s)=V_{0}^{2,\pi^{1,*},\pi^{2,*}}(s).
\]
Plugging the above results for $V^{k,*}(s)$ into the pair of equations
we obtained previously, we get that the pair of strategies $(\pi^{1,*},\pi^{2,*})$
satisfy
\[
\gamma V_{0}^{1,\pi^{1,*},\pi^{2,*}}(s)=\sup_{\pi^{1}\in\Pi^{1}}\bigl(r^{1}(s,\pi^{1},\pi^{2,*})+\sum_{s^{\prime}\in\mathcal{S}}\lambda(s^{\prime}|s,\pi^{1},\pi^{2,*})V_{0}^{1,\pi^{1,*},\pi^{2,*}}(s^{\prime})\bigr),
\]
and
\[
\gamma V_{0}^{2,\pi^{1,*},\pi^{2,*}}(s)=\sup_{\pi^{2}\in\Pi^{2}}\bigl(r^{2}(s,\pi^{1,*},\pi^{2})+\sum_{s^{\prime}\in\mathcal{S}}\lambda(s^{\prime}|s,\pi^{1,*},\pi^{2})V_{0}^{2,\pi^{1,*},\pi^{2,*}}(s^{\prime})\bigr).
\]
Then by the uniqueness of the solutions for the pair of optimality
equations (\ref{cont V1 game Bellman equation}) and (\ref{cont V2 game Bellman equation}),
we conclude that $(\pi^{1,*},\pi^{2,*})=(\pi_{0}^{1,*},\pi_{0}^{2,*})$.
\end{proof}

\subsection{Proof of Theorem \ref{Thm:2player Nash converge}}
\begin{proof}
To prove Theorem \ref{Thm:2player Nash converge}, it suffices to
show that our models can be reduced to the cases studied in the last
section.

By our construction, the state process $S_{t}$ (resp. $S_{i}$) in
the continuous-time (resp. discrete-time) game $\mathcal{G}_{0}$
(resp. $\mathcal{G}_{\Delta}$) is controlled Markov chain on finite
space $\mathcal{S}_{X}=\{\frac{k}{2}\delta_{P}\text{ }|\text{ }k=1,2,\ldots,(2N_{P}-1)\}$.
The rate parameters for $S_{t}$ are uniformly bounded. Then, we show
that the objective functions in $\mathcal{G}_{0}$ and $\mathcal{G}_{\Delta}$
can be expressed as the forms in Section \ref{sec:general MDP convergence results}.
For simplicity of notations, we omit the conditions $S_{0}=s$ in
the conditional expectations. All the following computations involving
the interchange between expectation and infinite series or integrals
are rigorous owing to the uniformly boundness of the functions therein.
By the properties of the continuous-time Markov chain and the expectation
rules for the stochastic integral w.r.t. the Poisson process (e.g.,
Section 3.2 in \citet{hanson2007BookJumpDiffusion}), we have that
\begin{align*}
V_{0}^{k,\pi^{1},\pi^{2}}(s) & =E[\int_{0}^{+\infty}e^{-\gamma t}(p_{t}^{a,k}-X_{t}-c)dN_{t}^{a,k}+\int_{0}^{+\infty}e^{-\gamma t}(X_{t}-p_{t}^{b,k}-c)dN_{t}^{b,k}]\\
 & =E[\int_{0}^{+\infty}e^{-\gamma t}(p_{t}^{a,k}-X_{t}-c)\Gamma^{a,k}(X_{t},p_{t}^{a,1},p_{t}^{a,2})dt+\int_{0}^{+\infty}e^{-\gamma t}(X_{t}-p_{t}^{b,k}-c)\Gamma^{b,k}(X_{t},p_{t}^{b,1},p_{t}^{b,2})dt]\\
 & =E[\int_{0}^{+\infty}e^{-\gamma t}r^{k}(X_{t},a_{t}^{1},a_{t}^{2})dt],
\end{align*}
where the reward function is defined as
\[
r^{k}(x,a^{1},a^{2}):=(p^{a,k}-x-c)\Gamma^{a,k}(x,p^{a,1},p^{a,2})+(x-p^{b,k}-c)\Gamma^{b,k}(x,p^{b,1},p^{b,2}).
\]
By the transition probabilities of the discrete-time Markov chains
in the MDP $\mathcal{M}_{\Delta}$, we have that
\[
V_{\Delta}^{k,\pi^{1},\pi^{2}}(s)=E[\sum_{i=0}^{+\infty}e^{-i\gamma\Delta}((p_{i}^{a,k}-X_{i}-c)n_{i}^{a,k}+(X_{i}-p_{i}^{b,k}-c)n_{i}^{b,k})]=E[\sum_{i=0}^{+\infty}e^{-i\gamma\Delta}r^{k}(X_{i},a_{i}^{1},a_{i}^{2})\Delta].
\]
Under Assumption \ref{Assump_2player_intensity}, we have that $r^{k}$
is uniformly bounded on the finite space $\mathcal{S}\times\mathcal{A}\times\mathcal{A}$.

Thus, the results in Theorem \ref{Thm:2player Nash converge} follow
from the general results in Lemma \ref{Lemma:Game converge general}.
\end{proof}

\end{document}